\documentclass[runningheads]{llncs}

\usepackage[T1]{fontenc}

\usepackage{booktabs}
\usepackage{subcaption}
\usepackage{amsmath}

\usepackage{amssymb}
\usepackage{array,makecell}
\usepackage[utf8]{inputenc}
\usepackage{microtype}
\usepackage{balance}
\usepackage{multirow}
\usepackage{diagbox}

\usepackage{graphics}
\usepackage{multicol}
\usepackage{enumitem}
\usepackage{mathtools}
\usepackage{epsfig}
\usepackage{xspace}
\usepackage{extarrows}
\usepackage{bm}

\usepackage{rotating}
\usepackage{eucal}
\usepackage[symbol, flushmargin]{footmisc}
\usepackage{multirow}
\usepackage{graphicx}
\usepackage{float}
\usepackage{ragged2e}
\usepackage[normalem]{ulem}
\usepackage{algorithm}
\usepackage{algorithmic}
\usepackage{setspace}
\usepackage{listings}
\usepackage{bbding}
\usepackage[most]{tcolorbox}
\lstdefinelanguage{program}
{
morekeywords={if, then, else, fi, while, true, false, switch, case, skip, break},
sensitive = false
}

\newcommand{\LTS}{ATS}

\newcommand{\tsVars}{X}
\newcommand{\tsVar}{x}
\newcommand{\tsLocs}{L}
\newcommand{\ToolName}{DInvG}
\newcommand{\tsLoc}{\ell}
\newcommand{\tsTrans}{\mathcal{T}}
\newcommand{\tsTran}{\tau}
\newcommand{\tsInitcond}{\theta}
\newcommand{\tsGuardcond}{\rho}
\newcommand{\tsEval}{\sigma}
\newcommand{\Rset}{\mathbb{R}}
\newcommand{\tsAssertphi}{\varphi}
\newcommand{\tsAssertpsi}{\psi}
\newcommand{\tsPath}{\pi}
\newcommand{\tsMap}{\eta}

\newcommand{\fkCoeff}{c}
\newcommand{\fkConst}{d}
\newcommand{\ProcSmry}{S}

\newcommand{\smalltitle}[1]{\vspace{0.5em}\noindent\textbf{#1}}
\newcommand{\Comment}[1]{\hfill \textcolor{gray}{\(\triangleright\) #1}}
\newcounter{counter1}
\newcounter{counter2}
\usepackage{hyperref}
\begin{document}

\title{Affine Disjunctive Invariant Generation \texorpdfstring{\\}{ }with Farkas' Lemma}

%\author{anonymous authors}
%\institute{ }


\author{Jingyu Ke\inst{1}\orcidID{0009-0008-6848-3105} \and 
Hongfei Fu\inst{1}\Envelope\orcidID{0000-0002-7947-3446} \and
Hongming Liu\inst{1}\orcidID{0000-0001-8987-596X} \and 
Zhouyue Sun\inst{1}\orcidID{0009-0009-8863-5548} \and 
Liqian Chen\inst{2}\orcidID{0000-0001-8084-8009} \and 
Guoqiang Li\inst{1}\orcidID{0000-0001-9005-7112}}

\authorrunning{J. Ke \emph{et al.}}

\institute{Shanghai Jiao Tong University, China\\
\email{\{Windocotber, jt002845, hm-liu, sunzhouyue, li.g\}@sjtu.edu.cn} \and 
National University of Defense Technology, China\\
\email{lqchen@nudt.edu.cn}}

\maketitle

\begin{abstract}
In the verification of loop programs, disjunctive  invariants are essential to capture complex loop dynamics such as phase and mode changes. In this work, we develop a novel approach for the automated generation of affine disjunctive invariants for affine while loops via Farkas' Lemma, a fundamental theorem on linear inequalities. 
Our main contributions are two-fold. First, we combine Farkas' Lemma with a succinct control flow transformation to derive disjunctive invariants from the conditional branches in the loop. 
Second, we propose an invariant propagation technique that minimizes the invariant  computation effort by propagating previously solved invariants to yet unsolved locations as much as possible. 
Furthermore, we resolve the infeasibility checking in the application of Farkas' Lemma which has not been addressed previously, and extend our approach to nested loops via loop summary.  
Experimental evaluation over more than 100 affine while loops (mostly from SV-COMP 2023) demonstrates that our approach is promising to generate tight linear invariants over affine programs. 
\end{abstract}

\section{Introduction}

An invariant at a program location is an assertion that over-approximates the set of program states reachable to that location, i.e., every reachable program state to the location is guaranteed to satisfy the assertion. Since invariants provide an over-approximation for reachable program states, they play a fundamental role in program verification and can be used for safety~\cite{DBLP:books/daglib/0080029,DBLP:conf/pldi/PadonMPSS16,DBLP:conf/cav/AlbarghouthiLGC12}, reachability~\cite{DBLP:conf/tacas/ColonS01,DBLP:conf/cav/BradleyMS05,DBLP:conf/sas/AliasDFG10,DBLP:conf/vmcai/PodelskiR04,DBLP:conf/ictac/ChenXYZZ07,DBLP:conf/fm/DavidKKL16,DBLP:conf/pldi/AsadiC0GM21} and time-complexity~\cite{DBLP:journals/toplas/ChatterjeeFG19} analysis.
Invariant generation targets the automated  generation of invariants which can be used to aid the verification of critical program properties.

Automated approaches for invariant generation have been studied for decades and there have been an abundance of literature along this line of research. From different program objects, invariant generation targets numerical values (e.g., integers or real numbers)~\cite{DBLP:conf/cav/ColonSS03,DBLP:conf/pldi/Chatterjee0GG20,DBLP:conf/sas/Rodriguez-CarbonellK04,DBLP:conf/popl/SinghPV17,DBLP:conf/sas/BagnaraHRZ03,DBLP:conf/vmcai/BoutonnetH19}, arrays~\cite{DBLP:conf/vmcai/LarrazRR13,DBLP:conf/pldi/SrivastavaG09}, pointers~\cite{DBLP:conf/pldi/LeZN19,DBLP:journals/jacm/CalcagnoDOY11}, algebraic data types~\cite{DBLP:journals/pacmpl/KSG22}, etc. By different methodologies, invariant generation can be solved by abstract interpretation~\cite{DBLP:conf/popl/CousotC77,DBLP:conf/popl/CousotH78,DBLP:conf/vmcai/BoutonnetH19,DBLP:conf/sas/GopanR07}, constraint solving~\cite{DBLP:conf/vmcai/Cousot05,DBLP:conf/cav/ColonSS03,DBLP:conf/pldi/GulwaniSV08,DBLP:conf/pldi/Chatterjee0GG20}, inference~\cite{DBLP:conf/oopsla/DilligDLM13,DBLP:journals/jacm/CalcagnoDOY11,DBLP:journals/fmsd/SharmaA16,DBLP:conf/cav/0001LMN14,DBLP:conf/sigsoft/Xu0W20,DBLP:conf/cav/GanX0ZD20,DBLP:conf/fmcad/SomenziB11,DBLP:conf/tacas/McMillan08,DBLP:conf/sas/DonaldsonHKR11}, recurrence analysis~\cite{DBLP:conf/fmcad/FarzanK15,DBLP:conf/pldi/KincaidBBR17,DBLP:journals/pacmpl/KincaidCBR18}, machine learning~\cite{DBLP:conf/popl/0001NMR16,DBLP:conf/pldi/HeSPV20,DBLP:conf/pldi/YaoRWJG20,DBLP:conf/iclr/RyanWYGJ20}, 
data-driven approaches~\cite{DBLP:conf/pldi/LeZN19,DBLP:conf/esop/0001GHALN13,DBLP:conf/icse/NguyenKWF12,DBLP:conf/icse/CsallnerTS08,DBLP:conf/cav/ChenHWZ15,FSE2022}, etc. Most results in the literature consider a strengthened version of invariants, called \emph{inductive invariants}, that requires the inductive condition that the invariant at a program location is preserved upon every execution back and forth to the location.

In this work, we consider the automated generation of disjunctive invariants, i.e., invariants that are in the form of a disjunction of assertions. Compared with conjunctive invariants, disjunctive invariants capture disjunctive features such as multiple phases and mode transitions in loops. Although extensive research has been conducted on conjunctive invariant generation, verification of programs with complex disjunctive loops still demands a more precise and scalable approach, rather than merely generating conjunctive invariants at the loop entry point to summarize the entire loop.

Moreover, most of the existing disjunctive invariant analyses rely on specific program patterns, such as alternating loop paths~\cite{DBLP:conf/cav/SharmaDDA11} or periodic regular loops~\cite{DBLP:conf/sigsoft/XieCLLL16,DBLP:conf/tase/LinZCSXLS21}, which cannot be effectively generalized to real-world programs with arbitrary execution traces. Methods based on abstract interpretation or symbolic execution often fail to converge, or converge to low-precision invariants, when dealing with loops that are deeply nested or have large constant iteration counts. To address this, we propose a modular approach capable of handling complex loops and efficiently mitigating the computational explosion caused by the interleaving of paths in disjunctive loops.

We use constraint solving to generate real-valued affine disjunctive invariants. 
A typical constraint solving method is via Farkas' Lemma~\cite{DBLP:conf/cav/ColonSS03,DBLP:conf/sas/SankaranarayananSM04,DBLP:conf/cav/JiFFC22,oopsla22/scalable} 
that provides a complete characterization for affine invariants. However, the application of Farkas' Lemma is mostly limited to the conjunction of affine inequalities. The question on how to leverage Farkas' Lemma to affine disjunctive invariants remains to be a challenge. In this paper, we focus on the generation of affine disjunctive invariants in affine loops. 
An affine loop is a while loop where all conditional and assignment statements are in the form of linear expressions.

\smallskip
\noindent{\em Our contributions.} First,
we introduce a novel control flow transformation that extracts loop paths (from entry to exit) as standalone locations in a transition system and establishes transitions between them. 
Second, to alleviate the exponential computational overhead introduced by the control flow transformation, we propose an invariant propagation technique that propagates already-computed invariants to locations whose invariants yet need to be computed as much as possible.
Third, we fully resolve the infeasible situation in the application of Farkas' Lemma~\cite{DBLP:conf/sas/SankaranarayananSM04,oopsla22/scalable} and extend our approach to nested loops through loop summary. 
Fourth, we implement our approach as a prototype \ToolName\footnote{The tool implementation is available on GitHub: \url{https://github.com/WindOctober/DInvG}.}. Experimental evaluation with various state-of-the-art verification tools using over $100$ benchmarks from SV-COMP 2023~\cite{svcomp} and \cite{DBLP:conf/vmcai/BoutonnetH19}, shows that our approach is both tight (in the accuracy of the generated invariant) and is time efficient for real-valued disjunctive affine invariant generation. 

The remainder of this paper is structured as follows. Sec.~\ref{sec:Preliminary} revisits the fundamental definitions of affine transition systems and invariants, providing a variant of \textit{Farkas' Lemma} as well as the basic definitions of polyhedra. Sec.~\ref{sec:overview} offers an overview of how our tool \ToolName\ transforms programs and solves for disjunctive invariants. In Sec.~\ref{sec:alg}, we formalize the definition of control flow transformation and extract the corresponding affine transition systems, present the pseudocode for invariant propagation as well as optimizations for nested loop and infeasible traces, and prove that the disjunctive invariants generated are inductive. Sec.~\ref{sec:exp} demonstrates the efficiency and precision advantages of \ToolName\ compared to several state-of-the-art tools and conducts an ablation study on invariant propagation. Sec.~\ref{sec:related_work} compares our method with related verification approaches, elaborating on the conceptual and implementation differences between invariant propagation and control flow transformation.
\section{Preliminaries}
\label{sec:Preliminary}
Below we revisit affine transition systems~\cite{DBLP:conf/sas/SankaranarayananSM04} and their associated invariants, elucidate Farkas' Lemma, and outline fundamental principles from polyhedra theory. It is important to note that, within the scope of this paper, we treat linear and affine concepts equivalently.

\subsection{Affine Transition Systems and Invariants}

An \emph{affine inequality} over a set $V=\{x_1,\dots,x_n\}$ of real-valued variables is of the form $a_{1}x_{1} + \dots + a_{n}x_{n} + b \ge 0$, where $a_i$'s and $b$ are real coefficients. 
An \emph{affine assertion} over $V$ is a conjunction of affine inequalities over $V$. 

An affine transition system possesses a finite number of locations as well as real-valued variables, and specifies transitions between locations with affine guards and affine updates on the values of the variables. 

\begin{definition}[Affine Transition Systems~\cite{DBLP:conf/sas/SankaranarayananSM04}]\label{def:lts}
An \emph{affine transition system} (\LTS) is a tuple $\Gamma=\langle \tsVars, \tsVars', \tsLocs, \tsTrans, \tsLoc^*, \tsInitcond \rangle$:
\begin{itemize}
    \item 
    $\tsVars$ is a finite set of real-valued variables and $\tsVars'=\{\tsVar' \mid \tsVar \in \tsVars\}$ is the set of primed variables.
    \item 
    $\tsLocs$ is a finite set of \emph{locations} and $\tsLoc^*\in\tsLocs$ is the initial location.
    \item 
    $\tsTrans$ is a finite set of \emph{transitions} where each transition $\tsTran$ is a triple $\langle \tsLoc, \tsLoc', \tsGuardcond \rangle$ from location $\tsLoc$ to location $\tsLoc'$ with the guard affine assertion $\tsGuardcond$ over $\tsVars \cup \tsVars'$.
    \item
    $\tsInitcond$ is a disjunction of affine assertions over $\tsVars$ that specifies the initial condition at $\tsLoc^*$.  
\end{itemize}
The directed graph $\mbox{\sl DG}(\Gamma)$ of the \LTS{} $\Gamma$ is defined as the graph where the vertices are the locations of $\Gamma$ and there is an edge $(\tsLoc,\tsLoc')$ if and only if there is a transition $\langle \tsLoc, \tsLoc', \tsGuardcond \rangle$ with source location $\tsLoc$ and target location $\tsLoc'$.
\end{definition}

The intuition of an ATS $\Gamma=\langle \tsVars, \tsVars', \tsLocs, \tsTrans, \tsLoc^*, \tsInitcond \rangle$ is as follows. 
Each variable $\tsVar \in \tsVars$ represents the current value of the variable and each primed variable $\tsVar' \in \tsVars'$ represents the next value of its unprimed variable $\tsVar \in \tsVars$ after one step of transition. 
The transition $\langle \tsLoc, \tsLoc', \tsGuardcond \rangle$ specifies the jump from the current location $\tsLoc$ to the next location $\tsLoc'$ with the guard condition $\tsGuardcond$ specifying the condition to enable the transition. The guard condition involves both the current values (represented by $\tsVars$) and the next values (by $\tsVars'$), so that it can specify the relationship between the current and next values.  

Below we describe the semantics of an ATS. 
A \emph{valuation} over a finite set $V$ of variables is a function $\tsEval : V \rightarrow \Rset$ that assigns to each variable $\tsVar \in V$ a real value $\tsEval(\tsVar) \in \Rset$. 
We mostly consider valuations over the variables $\tsVars$ of an \LTS{} and simply abbreviate ``valuation over $\tsVars$'' as ``valuation'' (i.e., omitting $\tsVars$). 
Given an \LTS, a \emph{configuration} is a pair $(\tsLoc, \tsEval)$ with the intuition that $\tsLoc$ is the current location and $\tsEval$ is a valuation that specifies the current values for the variables. 

Given an affine assertion $\tsAssertphi$ and a valuation $\tsEval$ over a variable set $V$, we write $\tsEval \models \tsAssertphi$ to mean that $\tsEval$ satisfies $\tsAssertphi$, i.e., $\tsAssertphi$ is true when one substitutes the corresponding values $\tsEval(\tsVar)$ into all the variables $\tsVar$ in $\tsAssertphi$. Given an \LTS{} $\Gamma$, two valuations $\tsEval,\tsEval'$ and an affine assertion $\tsAssertphi$ over $\tsVars \cup \tsVars'$, we write $\tsEval, \tsEval' \models \tsAssertphi$ to mean that $\tsAssertphi$ is true when one substitutes every variable $\tsVar \in \tsVars$ by $\tsEval(\tsVar)$ and every variable $\tsVar' \in \tsVars'$ with $\tsEval'(\tsVar)$ in $\tsAssertphi$. Moreover, given two affine assertions $\tsAssertphi,\tsAssertpsi$ over a variable set $V$, we write $\tsAssertphi \models \tsAssertpsi$ to mean that $\tsAssertphi$ implies $\tsAssertpsi$, i.e., for every valuation $\tsEval$ over $V$ we have that $\tsEval \models \tsAssertphi$ implies $\tsEval \models \tsAssertpsi$. The case of disjunction of affine assertions is similar.

The semantics of an \LTS{} $\Gamma$ is given by its  paths. 
A \emph{path} $\tsPath$ of the \LTS{} $\Gamma$ is a finite sequence of configurations $(\tsLoc_0,\tsEval_0) \dots (\tsLoc_k,\tsEval_k)$ such that 
\begin{itemize}
    \item(\textbf{Initialization}) $\tsLoc_0 = \tsLoc^*$ and $\tsEval_0 \models \tsInitcond$, and 
    \item(\textbf{Consecution}) for every $0 \le j \le k-1$, there exists a transition $\tsTran = \langle \tsLoc, \tsLoc', \tsGuardcond \rangle$ such that $\tsLoc = \tsLoc_{j}$, $\tsLoc' = \tsLoc_{j+1}$ and $\tsEval_{j},\tsEval_{j+1} \models \tsGuardcond$. 
\end{itemize}

We say that a configuration $(\tsLoc, \tsEval)$ is \emph{reachable} if there exists a path  $(\tsLoc_0,\tsEval_0) \ldots$ $(\tsLoc_k,\tsEval_k)$ such that $(\tsLoc_k,\tsEval_k)=(\tsLoc, \tsEval)$. 
An \emph{invariant} at a location $\tsLoc$ of an \LTS{} is an assertion $\tsAssertphi$ such that for every path $\tsPath = (\tsLoc_0,\tsEval_0) \dots (\tsLoc_k,\tsEval_k)$ of the \LTS{} and each $0 \le i \le k$, it holds that $\tsLoc_i = \tsLoc$ implies $\tsEval_i \models \tsAssertphi$. An invariant $\tsAssertphi$ is \emph{affine} if $\tsAssertphi$ is an affine assertion over the variable set $\tsVars$, and is \emph{disjunctively affine} if $\tsAssertphi$ is a disjunction of affine assertions. 

In invariant generation, one often investigates a strengthened version of invariants called \emph{inductive invariants}. In this work, we present affine inductive invariants in the form of inductive affine assertion maps~\cite{DBLP:conf/cav/ColonSS03,DBLP:conf/sas/SankaranarayananSM04,oopsla22/scalable} as follows. 

An \emph{affine assertion map} (AAM) over an \LTS{} is a function $\tsMap$ that maps every location $\tsLoc$ of the \LTS{} to an affine assertion $\tsMap(\tsLoc)$ over the variables $\tsVars$. An AAM $\tsMap$ is called \emph{inductive} if the following holds:
\begin{itemize}
    \item (\textbf{Initialization})  $\tsInitcond \models \tsMap(\tsLoc^*)$;
    \item (\textbf{Consecution}) For every transition $\tsTran = \langle \tsLoc, \tsLoc', \tsGuardcond \rangle$, we have that $\tsMap(\tsLoc) \wedge \tsGuardcond \models \tsMap(\tsLoc')'$, where $\tsMap(\tsLoc')'$ is the affine assertion obtained by replacing every variable $\tsVar \in \tsVars$ in $\tsMap(\tsLoc')$ with its next-value counterpart $\tsVar' \in \tsVars'$.
\end{itemize}
By a straightforward induction on the length of a path under an \LTS, one could verify that every affine assertion in an inductive AAM is indeed an invariant.

\subsection{Farkas' Lemma and Polyhedra}

Farkas' Lemma~\cite{FarkasLemma} is a classical theorem in the theory of affine inequalities and previous results ~\cite{DBLP:conf/cav/ColonSS03,DBLP:conf/sas/SankaranarayananSM04,oopsla22/scalable} have applied the theorem to affine invariant generation. In these results, the form of Farkas' Lemma follows~\cite[Corollary 7.1h]{DBLP:books/daglib/0090562}.

\begin{theorem}[Farkas' Lemma]\label{thr:farkas}
Consider an affine assertion $\tsAssertphi$ over a set $V=\{x_1,\dots, x_n\}$ of real-valued variables 
as in Figure~\ref{tsAssertphi in Farkas' Lemma}.
\noindent When $\tsAssertphi$ is satisfiable (i.e., there is a valuation over $V$ that satisfies $\tsAssertphi$), it implies an affine inequality $\tsAssertpsi$ as in Figure~\ref{tsAssertpsi in Farkas' Lemma} (i.e., $\tsAssertphi \models \tsAssertpsi$) if and only if there exist non-negative real numbers $\lambda_0, \lambda_1, \dots, \lambda_m$ such that 
(i) $c_j = \sum^{m}_{i=1} \lambda_{i} \cdot a_{ij}$ for all  $1 \le j \le n$, and 
(ii) $d = \lambda_0 + \sum^{m}_{i=1} \lambda_{i} \cdot b_{i}$ as in Figure~\ref{tab:farkas}. 
Moreover, $\tsAssertphi$ is unsatisfiable if and only if the inequality $-1 \geq 0$ (as $\tsAssertpsi$) can be derived from above.
\end{theorem}

\begin{figure}[h]
    \begin{minipage}[b]{0.4\linewidth}
        \begin{minipage}[b]{\linewidth}
            \centering
            \setlength{\arraycolsep}{0mm}{\small 
                $\tsAssertphi : $
                $\begin{array}{ccccccccccccc}
                    a_{11} & \cdot  & x_1 & + & \cdots & + & a_{1n} & \cdot  & x_n & + & b_1    & \geq & 0 \\
                           & \vdots &     &   &        &   &        & \vdots &     &   & \vdots &      &   \\
                    a_{m1} & \cdot  & x_1 & + & \cdots & + & a_{mn} & \cdot  & x_n & + & b_m    & \geq & 0    
                \end{array}$
            }
            \subcaption{$\tsAssertphi$ in Farkas' Lemma}
            \label{tsAssertphi in Farkas' Lemma}
        \end{minipage}
        \begin{minipage}[b]{\linewidth}
            \centering
            \setlength{\arraycolsep}{0mm}{
                \vspace{1.5ex}
                $\tsAssertpsi : $
                $\begin{array}{ccccccccccccc}
                    c_{1} & \cdot  & x_1 & + & \cdots & + & c_{n} & \cdot  & x_n & + & d    & \geq & 0 \\
                \end{array}$
            }
            \subcaption{$\tsAssertpsi$ in Farkas' Lemma}
            \label{tsAssertpsi in Farkas' Lemma}
        \end{minipage}
    \end{minipage}
    %\hfill
    \begin{minipage}[b]{0.58\linewidth}
        \setlength{\arraycolsep}{0mm}{\small
        $$
        \begin{array}{c|rcccrcrcc}
            \lambda_{0}  &            &     &         &     &            &     &1      &\geq &0  \\ 
            \lambda_{1}  &a_{11}\cdot x_1 &+  &\cdots &+  &a_{1n}\cdot x_n &+  &b_1    & \geq &0  \\ 
            \vdots         &\vdots    &     &         &     &\vdots    &     &\vdots &       &         \\ 
            \lambda_{m}  &a_{m1}\cdot x_1 &+  &\cdots &+  &a_{mn}\cdot x_n &+  &b_m    &\ge &0  \\ 
            \hline
                             &c_1\cdot x_1    &+  &\cdots &+  &c_n\cdot x_n    &+  &d      & \geq &0  \\
                             &            &     &         &     &            &     &-1     &\geq &0
        \end{array}
        \begin{array}{cc}
              &  \\
             \left.
            \begin{array}{c}
                \vspace{1.5ex} \\
                    \\ 
                    \\
             \end{array} 
             \right\}&\tsAssertphi \\
             \leftarrow  &\tsAssertpsi \\
             \leftarrow  &\mbox{\sl false}
        \end{array}
        $$
        }
        \subcaption{The Tabular Form for Farkas' Lemma}
        \label{tab:farkas}
    \end{minipage}
    \caption{The $\tsAssertphi$, $\tsAssertpsi$ and Tabular Form for Farkas' Lemma~\cite{DBLP:conf/cav/ColonSS03,DBLP:conf/sas/SankaranarayananSM04}}
    \label{The Text Form and Tabular Form for Farkas' Lemma}
\end{figure}

Farkas' Lemma simplifies the inclusion of a polyhedron inside a halfspace into the satisfiability of a system of affine inequalities. We refer to the case of unsatisfiable $\tsAssertphi$ with $\tsAssertpsi:=-1\ge 0$ in the statement of Theorem~\ref{thr:farkas} as \emph{infeasible implication}. 
The application of Farkas' Lemma can be visualized by the tabular form in Figure~\ref{tab:farkas} (taken from~\cite{DBLP:conf/cav/ColonSS03}), and we multiply $\lambda_0,\lambda_1,\dots,\lambda_m$ with their inequalities in $\tsAssertphi$ and sum up them together to get $\tsAssertpsi$. For $1 \le j \le m$, we require $\lambda_j \ge 0$. 

A subset $P$ of $\Rset^n$ is a \emph{polyhedron} if $P=\{\mathbf{x}\in \Rset^n\mid \mathbf{A}\cdot\mathbf{x}\le \mathbf{b}\}$ for some real matrix $A\in \Rset^{m\times n}$ and real vector $\mathbf{b}\in \Rset^m$, where $\mathbf{x}$ is treated as a column vector and the comparison $\mathbf{A}\cdot\mathbf{x}\le \mathbf{b}$ is defined in the coordinate-wise fashion. 
A polyhedron $P$ is a \emph{polyhedral cone} if $P=\{\mathbf{x}\in \Rset^n\mid \mathbf{A}\cdot\mathbf{x}\le \mathbf{0}\}$ for some real matrix $A\in \Rset^{m\times n}$, 
where $\mathbf{0}$ is the $m$-dimensional zero column vector. 
It is well-known from the Farkas-Minkowski-Weyl Theorem~\cite[Corollary 7.1a]{DBLP:books/daglib/0090562} that any polyhedral cone $P$ can be represented as $P=\{\sum_{i=1}^k \lambda_i\cdot \mathbf{g}_i\mid \lambda_i\ge 0 \mbox{ for all }1\le i\le k\}$ for some real vectors $\mathbf{g}_1,\dots,\mathbf{g}_k$, where such vectors $\mathbf{g}_i$'s are called a collection of \emph{generators} for the polyhedral cone $P$. 

\section{An Overview of Our Approach}
\label{sec:overview}
Consider the affine program \( P_1 \) in Figure~\ref{fig:ExampleLoopAndTransformedLoop}. Our approach has three parts, namely control flow transformation, invariant computation and invariant propagation.  

\smallskip
\noindent\emph{Control Flow Transformation}. For the non-nested loop \(P_1\), we extract each execution path from the loop entry to the exit and transform it into the form of loop \(P_2\). Each \textbf{case} statement corresponds to a possible path in the original loop with a path condition $\phi$ of taking the path specified by 
the conjunction of a conditional formula $\phi_c$  and an assignment formula $\phi_a$. For example, in the first case of \(P_2\) that corresponds to the case of entering the if-branch in \(P_1\), we have \(\phi_a\) is 
\((x' = x + 1 \land y' = y + 1)\), \(\phi_c\) is \(x > 49\), and \(\phi\) is \(\phi_a \land \phi_c\).

\begin{figure}[t]
    \centering
    \begin{minipage}[t]{0.4\textwidth}
        \centering
        \includegraphics[width=\linewidth]{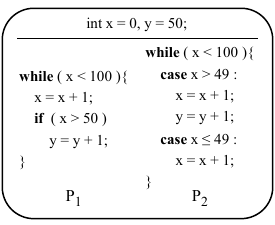}
        \subcaption{Source $P_1$ and Its Transformation $P_2$.}
        \label{fig:ExampleLoopAndTransformedLoop}
    \end{minipage}
    \hfill
    \begin{minipage}[t]{0.54\linewidth}
        \centering
        \vspace{-4.1cm}
        \begin{small}
            $\tsVars = \{ x, y \}$,
            $\tsLocs = \{ \tsLoc_{1}, \tsLoc_{2}^{*}\}$,
            $\tsTrans = \{ \tsTran_{1}, \tsTran_{2}, \tsTran_{3}, \tsTran_{4} \}$,
            $\tsInitcond : x = 0 \wedge y = 50$, 
            $\tsTran_{1} : \langle \tsLoc_{1}, \tsLoc_{1}, \tsGuardcond_{1} \rangle$, 
            $\tsTran_{2} : \langle \tsLoc_{1}, \tsLoc_{2}, \tsGuardcond_{2} \rangle$, 
            $\tsTran_{3} : \langle \tsLoc_{2}, \tsLoc_{2}, \tsGuardcond_{3} \rangle$, 
            $\tsTran_{4} : \langle \tsLoc_{2}, \tsLoc_{1}, \tsGuardcond_{4} \rangle$, 
            $$
            \tsGuardcond_{1} : \left[ 
            \setlength{\arraycolsep}{0mm}{
                    \begin{array}{rcl}
                    50 \le x  \le 99 \vspace{-0.5ex}\\ 
                    50 \le x' \le 99 \vspace{-0.5ex}\\ 
                    x' = x + 1 \vspace{-0.5ex}\\ 
                    y' = y + 1 
                \end{array}
            }
            \right],
            \tsGuardcond_{2} : \left[ 
            \setlength{\arraycolsep}{0mm}{
                \begin{array}{rcl}
                    50 \le x  \le 99 \vspace{-0.5ex}\\ 
                    x' \le 49 \vspace{-0.5ex}\\ 
                    x' = x + 1 \vspace{-0.5ex}\\ 
                    y' = y + 1 
                \end{array}
            }
            \right]
            $$\vspace{-1ex}
            $$
            \tsGuardcond_{3} : \left[ 
            \setlength{\arraycolsep}{0mm}{
                \begin{array}{rcl}
                    x  \le 49 \vspace{-0.5ex}\\ 
                    x' \le 49 \vspace{-0.5ex}\\ 
                    x' = x + 1 \vspace{-0.5ex}\\ 
                    y' = y 
                \end{array}
            }
            \right],
            \tsGuardcond_{4} : \left[ 
            \setlength{\arraycolsep}{0mm}{
                \begin{array}{rcl}
                    x  \le 49 \vspace{-0.5ex}\\ 
                    50 \le x' \le 99 \vspace{-0.5ex}\\ 
                    x' = x + 1 \vspace{-0.5ex}\\ 
                    y' = y 
                \end{array}
            }
            \right]
            $$
        \end{small}
    \subcaption{The \LTS{} Corresponding to $P_2$}
    \label{fig:ATSforTransformedLoop}
    \end{minipage}

    \caption{An example from~\cite{DBLP:conf/cav/SharmaDDA11} and its transformed form and corresponding \LTS{}}
    \vspace{-1em}
    \label{fig:ExampleCodeAndATS}
\end{figure}

Then, an ATS in Figure~\ref{fig:ATSforTransformedLoop} is directly derived from the transformed program. In this example, each \textbf{case} statement is considered as an independent location in the ATS, e.g., the location $\tsLoc_1$ stands for the first \textbf{case} in $P_2$. The transitions between the locations are derived from the jumps between different paths, where each transition guard $\tsGuardcond$ for a transition $\tsTran$, can be obtained through the formula:
$$
\tsGuardcond:= \phi_c \wedge \phi'_c[x'/x] \wedge \phi_a\wedge G\wedge G[x'/x]
$$
where, \(\phi_a\) and \(\phi_c\) represent the assignment and conditional formulas at the transition's start location, \(\phi'_c\) represents the conditional formula at the transition's end location, and \(G\) is the loop guard. After the \LTS{} is constructed, we apply the approaches~\cite{DBLP:conf/sas/SankaranarayananSM04,oopsla22/scalable} in Farkas' Lemma combined with the technique of invariant propagation to obtain invariants at all locations, and group them disjunctively together to obtain result invariants for original loop.
Note that the construction of the ATS is the key connection to apply Farkas' Lemma. 

\smallskip
\noindent\emph{Invariant Computation.} We first establish affine invariant templates at each location in the ATS by setting: 
$$
\tsMap(\tsLoc_{i}) := \fkCoeff_{\tsLoc_{i},1}x + \fkCoeff_{\tsLoc_{i},2}y + \fkConst_{\tsLoc_{i}} \geq 0\ ,\forall i \in \{1,2\}
$$
where, $\fkCoeff_{\tsLoc_{i},1},\fkCoeff_{\tsLoc_{i},2},\fkConst_{\tsLoc_{i}}$ are unknown coefficients to be resolved. Then, we generate the constraints from the initialization as well as consecution conditions via the Farkas' tabular in Figure~\ref{tab:farkas} and derive the initialization tabular and consecution tabular as shown in Figure~\ref{tab:farkasinitcons}. 
Setting \(\varphi=\theta\) and \(\psi=\eta(\tsLoc^*)\) in Theorem~\ref{thr:farkas}, the initialization of $\tsInitcond \models \tsMap(\tsLoc^*)$ results in linear constraints, 
while the consecution condition $\tsMap(\tsLoc) \wedge \tsGuardcond \models \tsMap(\tsLoc')'$ by setting \(\varphi=\tsMap(\tsLoc) \wedge \rho\) and \(\psi=\eta(\tsLoc')'\) results in quadratic constraints since we have a fresh $\lambda$, which we denote as $\mu$, multiplied by $\tsMap(\tsLoc)$.

\begin{figure}[t]
    \begin{minipage}{\linewidth}
        \centering
        \setlength{\arraycolsep}{0mm}{
            $$
            \begin{array}{c|rcccrcrcc}
                \lambda_{0}  &            &     &         &     &            &     &1      &\geq &0  \\ 
                \lambda_{1}  &a_{11}x_1 &+  &\cdots &+  &a_{1n}x_n &+  &b_1    &\Join_1 &0  \\ 
                \vdots         &\vdots    &     &         &     &\vdots    &     &\vdots &       &         \\ 
                \lambda_{m}  &a_{m1}x_1 &+  &\cdots &+  &a_{mn}x_n &+  &b_m    &\Join_m &0  \\ 
                \hline
                                 &c_{\tsLoc^*,1} x_1    &+  &\cdots &+  &c_{\tsLoc^*,n} x_n    &+  &d_{\tsLoc^*}      &\geq &0  \\
                                 &            &     &         &     &            &     &-1     &\geq &0
            \end{array}
            \begin{array}{cc}
                  &  \\
                 \left.
                \begin{array}{c}
                    \vspace{1.5ex} \\
                        \\ 
                        \\
                 \end{array} 
                 \right\}&\theta   \\
                 \leftarrow  &\tsMap(\tsLoc^*) \\
                 \leftarrow  &\mbox{\sl false}
            \end{array}
            $$
        }
        \subcaption{Initialization Tabular}
        \label{tab:farkasinit}
    \end{minipage}
    \begin{minipage}{\linewidth}
        \centering
        \setlength{\arraycolsep}{0mm}{\small
        $$
        \begin{array}{c|rcccrcrcccrcrcc}
            \mu        &c_{\tsLoc, 1} x_1     &+  &\cdots &+ &c_{\tsLoc, n} x_n    &    &     &     &     &     &      &+  &d_{\tsLoc}  &\geq  &0  \\
        
            \lambda_{0}  &       &    &       &    &       &    &     &     &     &     &      &     &1  &\geq     &0      \\
            
            \lambda_{1}  &a_{11} x_1       &+ &\cdots &+ &a_{1n} x_n       &+                                     
                             &a_{11}' x_{1}'   &+ &\cdots &+ &a_{1n}' x_{n}'   &+  &b_1&\Join_1     &0      \\ 
            
            \vdots         &\vdots                &    &         &    &\vdots                &     
                             &\vdots                &    &         &    &\vdots                &     &\vdots         &   &      \\ 
            
            \lambda_{m}  &a_{m1} x_1       &+ &\cdots &+ &a_{mn} x_n       &+  
                             &a_{m1}' x_{1}'   &+ &\cdots &+ &a_{mn}' x_{n}'   &+  &b_m&\Join_m     &0      \\ 
            \hline
                             &   &   &   &   &  &   
                             &c_{\tsLoc', 1} x_{1}' &+  &\cdots &+ &c_{\tsLoc', n} x_{n}'  &+  &d_{\tsLoc'}  &\geq  &0         \\
                             
                             &       &    &       &    &       &    &     &     &     &     &      &     &-1 &\geq  &0
        \end{array}
        \begin{array}{cc}
            \leftarrow  &\tsMap(\tsLoc)  \\
                        &                \\
             \left.
            \begin{array}{c}
                   \vspace{2ex} \\
                    \\ 
                    \\
             \end{array} 
             \right\}&\rho   \\
             \leftarrow  &\tsMap(\tsLoc')' \\
             \leftarrow  &\mbox{\sl false}
        \end{array}
        $$
        }        
        \subcaption{Consecution Tabular}
        \label{tab:farkascons}
    \end{minipage}
    \caption{Tabular for Initialization and Consecution~\cite{oopsla22/scalable}}
    \label{tab:farkasinitcons}
\end{figure}

In this context, we adopt a unified notation by employing $\eta(\tsLoc)$ to represent both affine expressions and affine inequalities interchangeably throughout the manuscript. To resolve the quadratic constraints from the consecution, the previous approach~\cite{DBLP:conf/sas/SankaranarayananSM04} has considered heuristics to guess the value of the multiplier $\mu$ through either practical rules such as factorization or setting $\mu$ manually to $0$ or $1$. These settings relax the original consecution definition into two stronger forms:

\begin{itemize}
    \item \textbf{Local Consecution}: For transition $\tau : \langle \tsLoc_i, \tsLoc_j, \tsGuardcond \rangle$, \(\tsGuardcond \models \tsMap(\tsLoc_j)' \geq 0\),
    \item \textbf{Incremental Consecution}: For transition $\tau : \langle \tsLoc_i, \tsLoc_j, \tsGuardcond \rangle$, \(\tsGuardcond \models \tsMap(\tsLoc_j)' \geq \tsMap(\tsLoc_i)\).
\end{itemize}

Then, we collect the derived constraints to constitute a formula $\Phi$ in CNF, which further expands into a DNF \(\Phi'\). Note that each disjunctive clause in the DNF \(\Phi'\) is an affine assertion defining a polyhedral cone, and we solve for the unknown coefficients of the templates by computing the generators of these polyhedral cones. We present a polyhedral cone in the DNF \(\Phi'\) of the ATS in Figure~\ref{tab:geninv}, along with corresponding generators, as shown in Figure~\ref{tab:cs}, where "point" means a single vector, "ray" means a vector that can be scaled by an arbitrary positive value, and "line" means a vector that can be scaled by any positive or negative value. When putting the generators back to invariants, we obtain the invariants shown in the right part of Figure~\ref{tab:geninv}. The resulting disjunctive invariant is the disjunction of invariants at all locations over an \LTS{} \( \Gamma \), corresponding to the invariants required on different loop paths.

\smallskip
\noindent{\em Invariant Propagation.} In the preceding exposition, the computation of conjunctive affine invariants adheres to existing approaches~\cite{DBLP:conf/sas/SankaranarayananSM04,oopsla22/scalable}. However, the strategies employed to resolve invariants at each location across the entire \(\LTS\) result in a significant decline in computational efficiency. To address this limitation, we introduce a propagation technique predicated on existing invariant computation results:

\begin{figure}[t]
    \begin{minipage}[b]{0.35\linewidth}
        \centering
        {\small
        $$
        \left[
        \setlength{\arraycolsep}{0pt}{
        \begin{array}{r} 
            \fkCoeff_{12} - \fkCoeff_{22} = 0, \\ 
            \fkCoeff_{21} \geq 0, \\
            \fkCoeff_{11} + \fkCoeff_{12} \geq 0, \\ 
            
            50\fkCoeff_{12} + \fkConst_{2} \geq 0, \\ 
            50\fkCoeff_{11} + \fkConst_{1} - 49\fkCoeff_{21} - \fkConst_{2} \geq 0 
        \end{array}
        }
        \right]
        $$}
        \subcaption{A clause in the DNF}
        \label{tab:cs}
    \end{minipage}
    \hfill
    \begin{minipage}[b]{0.55\linewidth}
        \centering
        \resizebox{0.95\textwidth}{!}{
        \setlength{\tabcolsep}{0.5mm}{
        \begin{tabular}{ccccccc|cc}
        \hline
        type & $c_{11}$  &$c_{12}$  &$d_{1}$  &$c_{21}$  &$c_{22}$  &$d_{2}$   &$\tsMap(\tsLoc_{1})$ &$\tsMap(\tsLoc_{2})$\\ 
        \hline
        {\rm \ point}&$0$&$0$&$0$&$0$&$0$&$0$ &$0\geq0$&$0\geq0$\\ 
        {\rm \ line}&$1$&$-1$&$0$&$0$&$-1$&$50$ &$x-y=0$&$-y+50=0$\\ 
        {\rm \ ray}&$0$&$0$&$49$&$1$&$0$&$0$ &$49\geq0$&$x\geq0$\\ 
        {\rm \ ray}&$0$&$0$&$1$&$0$&$0$&$0$ &$1\geq0$&$0\geq0$\\ 
        {\rm \ ray}&$1$&$0$&$-50$&$0$&$0$&$0$ &$x-50\geq0$&$0\geq0$\\ 
        {\rm \ ray}&$0$&$0$&$1$&$0$&$0$&$1$&$1\geq0$&$1\geq0$\\ 
        \hline
        \end{tabular}}
        }
        \subcaption{generators and their invariants}
        \label{tab:geninv}
    \end{minipage}
    \caption{Example of a clause in the DNF with its generators and invariants}
    \vspace{-1.5em}
    \label{tab:csgeninv}
\end{figure}

\begin{figure}[h]
    \centering
    \begin{minipage}[h]{0.28\textwidth}
        \includegraphics[width=1.1\linewidth,trim={0 0.5cm 0 0.2cm}]{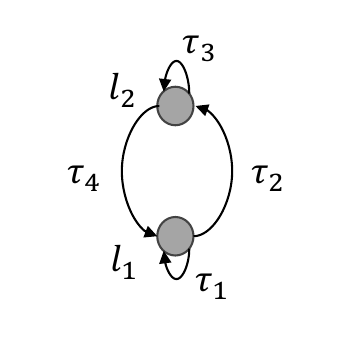}
        \subcaption{$DG(\Gamma)$}
        \label{fig:ExampleDirectedGraph}
        
    \end{minipage}
    \hfill
    \begin{minipage}[h]{0.35\textwidth}
        \includegraphics[width=\linewidth,trim={0 0.7cm 0 0.6cm}]{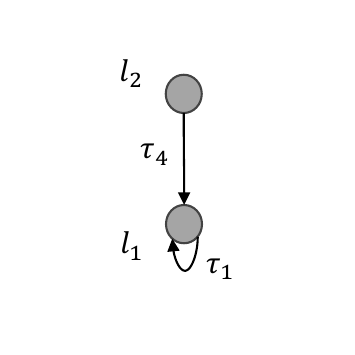}
        \subcaption{$DG(\Gamma)$ After Elimination}
        \label{fig:ExampleDirectedGraphEliminateInit} 
    \end{minipage}
    \hfill
    \begin{minipage}[h]{0.35\textwidth}
        \includegraphics[width=\linewidth,trim={0 0.7cm 0 0.6cm}]{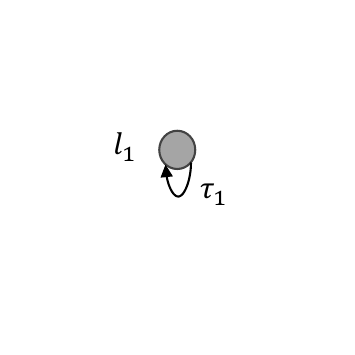}
        \subcaption{$DG(\Gamma)$ After Propagation}
        \label{fig:ExampleDirectedGraphAfterCompute} 
    \end{minipage}
    \vspace{-0.5em}
    \caption{Procedure of Invariant Propagation for example in Figure~\ref{fig:ExampleCodeAndATS}}
    \vspace{-1.5em}
    \label{fig:ExampleofInvPropagation}
\end{figure}

Consider the affine transition system $\Gamma$ in Figure~\ref{fig:ATSforTransformedLoop}. Its underlying directed graph $DG(\Gamma)$ is given in Figure~\ref{fig:ExampleDirectedGraph}. We first compute the invariant \(\tsMap(\tsLoc_{2}):= y = 50\wedge 0 \leq x \leq 49\) at the initial location \(\tsLoc_2\) as above. Then, we can eliminate all transitions pointing to \( \tsLoc_2 \) (the correctness of which is demonstrated in the following section), obtaining the graph in Figure~\ref{fig:ExampleDirectedGraphEliminateInit}. Notably, there exists a topological order, where \( \tsLoc_2\) precedes \( \tsLoc_1 \). Thus, we propagate the invariants at \( \tsLoc_2 \) along the transition \( \tsTran_4 \)'s transition guard $\tsGuardcond$ to establish the initial condition \(\tsInitcond\) of the \LTS{} over \( \tsLoc_1 \) in Figure~\ref{fig:ExampleDirectedGraphAfterCompute}, which is derived from removing \( \tsLoc_2 \) after propagation. Finally, by solving the invariants for the simplified ATS instead of the original ATS, we obtain the affine invariant \(\tsMap(\tsLoc_{1})= (x=y \wedge 50 \leq x \leq 99)\). Note that our invariant propagation is different from abstract interpretation, see Section~\ref{sec:related_work} for details.

Moreover, \(\eta(\tsLoc_{1})\) corresponds to the invariant within the loop, specifically representing the invariant at the first \textbf{case}. For the exit state \( x = 100 \), the disjunctive invariant generated within the loop, in conjunction with the loop exit condition \(\lnot G\), derives the program state \( x = 100\) outside the loop through stepwise deduction.

\section{Algorithmic Details in Our Approach}
\label{sec:alg}
Below we present our approach for generating affine disjunctive invariants over affine programs. We first illustrate our control flow transformation for non-nested loops, then our invariant propagation to reduce invariant computation, and finally the resolution of the infeasible implication and the extension to nested loops.

\subsection{Control Flow Transformation}\label{sec:cft}
We fix the set of program variables in a loop
as $X=\{x_1,\dots,x_n\}$ and identify it as the set of variables in the \LTS{} to be derived from the loop. We consider the canonical form of a non-nested affine while loop as in Figure~\ref{fig:CanonicalForm} similar to \cite{DBLP:conf/cav/JiFFC22}, where we have: 
\begin{itemize}
\item The column vector $\mathbf{x}=(x_1,\dots,x_n)^{\mathrm{T}}$ represents the vector of program variables, and $G$ is a disjunction of affine assertions that serves as the loop condition. 
\item  Each $\mathbf{F}_i$ ($1\le i\le m$) is an affine function, i.e.,  $\mathbf{F}_i(\mathbf{x})=\mathbf{A} \mathbf{x}+\mathbf{b}$ where $\mathbf{A}$ (resp. $\mathbf{b}$) is an $n\times n$ square matrix (resp. $n$-dimensional column vector) that specifies the affine update under the function $\mathbf{F}_i$ in the conditional branch $\phi_i$. 
The assignment $\mathbf{x}:=\mathbf{F}_i(\mathbf{x})$ is considered simultaneously for the variables in $\mathbf{x}$ so that in one execution step, the current valuation $\tsEval$ is updated to $\mathbf{F}_i(\tsEval)$. 
\item The statements $\delta_1,\dots,\delta_m$ specify whether the loop continues after the affine update of the conditional branches $\phi_1,\dots,\phi_m$. Each statement $\delta_i$ is either the \textbf{skip} statement that has no effect or the \textbf{break} statement that exits. 
\end{itemize}

\lstset{language=program}
\lstset{tabsize=3}
\newsavebox{\unnested}
\begin{lrbox}{\unnested}
\begin{lstlisting}[mathescape]
while $(G)$ {
  case $\phi_1$: $\mathbf{x}:=\mathbf{F}_1(\mathbf{x})$;$\delta_1$;
      $\vdots$ 
  case $\phi_m$: $\mathbf{x}:=\mathbf{F}_m(\mathbf{x})$;$\delta_m$;
}
\end{lstlisting}
\end{lrbox}

\begin{figure}[ht]
\begin{minipage}[b]{\linewidth} 
\centering
\resizebox{0.5\linewidth}{!}{\usebox{\unnested}}
\caption{The canonical form of a non-nested affine while loop}
\label{fig:CanonicalForm}
\end{minipage}
\end{figure}

Any non-nested affine while loop with a break statement can be transformed into the canonical form in Figure~\ref{fig:CanonicalForm} by recursively examining the substructures of the loop body. 
A detailed recursive transformation process is provided in Appendix~\ref{appendix:transform}.
Note that although the transformation into our canonical form may cause exponential growth in the number of conditional branches in the loop body, in practice a loop typically has a small number of conditional branches and further improvement can be carried out by removing invalid branches (i.e., those whose branch condition is unsatisfiable, such as $\tau_2$ in Figure~\ref{fig:ExampleCodeAndATS}). 

Moreover, such a canonical form is often necessary to capture precise disjunctive information in a while loop. Each case corresponds not merely to an individual loop path but also encapsulates the set of states at the entry of a loop. By finely partitioning the incoming program states according to branching conditions and formulating constraints among these states, we endeavor to precisely characterize the dynamics of internal state transitions within loops exhibiting multi-phase behavior.
Below we demonstrate our control flow transformation that transforms the canonical form into an ATS. 

Formally, the \LTS{} $\Gamma_W$ for a loop $W$ in our canonical form is 
%constructed 
given as follows:
\begin{itemize}
\item The set of locations is $\{\tsLoc_1,\dots, \tsLoc_m, \tsLoc_{e}\}$, 
where each $\tsLoc_i$ ($1\le i\le m$) corresponds to \(i\)-th case in the canonical form and $\tsLoc_{e}$ is the termination location of the loop. 
\item For each $1\le i\le m$ such that $\delta_i=\mathbf{break}$, we have the transition (we denote $\mathbf{x}':=(x'_1,\dots,x'_n)^\mathrm{T}$)
$$
\tau_{i}=(\tsLoc_i, \tsLoc_e, G \wedge \phi_i \wedge  \mathbf{x}'=\mathbf{F}_i(\mathbf{x}))
$$
that specifies the one-step jump from \(\tsLoc_i\) to the termination location $\tsLoc_e$.
\item For each $1\le i,j\le m$ such that $\delta_i\neq \mathbf{break}$, we have the transition 
$$
\tau_{ij}=(\tsLoc_i, \tsLoc_j, G \wedge \phi_i \wedge G[\mathbf{x}'/\mathbf{x}] \wedge \phi_j[\mathbf{x}'/\mathbf{x}] \wedge \mathbf{x}'=\mathbf{F}_i(\mathbf{x}))
$$ 
that specifies the jump from \(\tsLoc_i\) in the current loop iteration to \(\tsLoc_j\)  in the next loop iteration.
\item For each $1\le i\le m$ such that $\delta_i\neq \mathbf{break}$, we have the transition 
$$
\tau'_{i}=(\tsLoc_i, \tsLoc_{e}, G\wedge \phi_i\wedge (\neg G)[\mathbf{x}'/\mathbf{x}]\wedge \mathbf{x}'=\mathbf{F}_i(\mathbf{x}))
$$ 
for the jump from \(\tsLoc_i\) to the termination location $\tsLoc_e$.
\end{itemize} 

After the transformation, we remove transitions with an unsatisfiable guard condition to reduce the size of the derived \LTS{}. The transformation for our running example has been given in Figure~\ref{fig:ExampleCodeAndATS}. 
The transformed \LTS{} $\Gamma_W$ enables us to apply existing approaches~\cite{DBLP:conf/sas/SankaranarayananSM04,oopsla22/scalable} to generate invariants at the locations of the ATS $\Gamma_W$. 
Finally, recall that the overall disjunctive invariant for the ATS $\Gamma_W$ is the disjunction of the invariants at all the locations. 
    
As for the previous approaches~\cite{DBLP:conf/sas/SankaranarayananSM04,oopsla22/scalable} that only set locations at the initial location of the loop, our control flow transformation has different locations corresponding to different paths of the original loop, thereby achieving fine-grained piecewise invariants. 
Moreover, contrary to the granular program translations employed in traditional software model checking, which is similar to the control flow transformation, our motivation is mainly to integrate this transformation with Farkas' Lemma. By analyzing the transition patterns among internal loop paths, we aim to more effectively capture the phase-specific characteristics of multi-stage programs.

As demonstrated by our experiments, considering transitions between any pair of paths \((\tsLoc_i, \tsLoc_j)\), rather than partitioning the loop into more complex path regular expressions (using regular expressions to abstract loop behaviors over multiple iterations) allows us to maintain a balance between the precision of invariant generated and the efficiency of constraint solving. To further improve the efficiency, we have designed the invariant propagation algorithm shown below for these directed graphs constructed from paths.

\subsection{Invariant Propagation}
In the computation of invariants, previous approaches ~\cite{DBLP:conf/sas/SankaranarayananSM04,oopsla22/scalable} require to generate the invariants at all the locations of an \LTS{}. As invariant computation is usually expensive, it is important to explore optimizations that avoid redundant computations. In this section, we propose a novel invariant propagation technique that is applicable to any directed graph of an \LTS{} and achieves maximal efficiency in specific graph structures such as directed cycles. Below we  demonstrate the procedure of invariant propagation via Algorithm~\ref{alg:invprop}.

\begin{algorithm}[t]
    \caption{\textit{InvProp}($\Gamma$,\(DG(\Gamma)\),\(\tsLoc^*\))}
    \begin{algorithmic}[1]
        \REQUIRE $\Gamma$ --- \LTS{}, $DG(\Gamma)$ --- directed graph of $\Gamma$, $\tsLoc^*$ --- initial location of $\Gamma$.
        \ENSURE $\mathbf{\eta}$ --- an inductive assertion map for $\Gamma$.
        \STATE \textit{Init assertion map} \(\eta\) for \(\Gamma\) .  \label{line:1}
        \STATE $\text{SCCs} \gets \textit{Tarjan}(DG(\Gamma),\Gamma)$ \Comment{Find all SCCs in the directed graph} \label{line:2}
        \IF{$\textit{Len}(\text{SCCs}) \neq 1$} \label{line:3}
            \STATE $\text{id} \gets \textit{FindSCC}(\tsLoc^*,\text{SCCs})$ \Comment{Find the SCC containing $\tsLoc^*$} 
            \STATE $\text{stack}.\textit{push}(\text{id},\tsLoc^*)$ 
            \WHILE{$\neg \text{stack}.\textit{isEmpty}()$}
                \STATE $(\text{cur}, \tsLoc_s )\gets \text{stack}.\textit{pop}()$ \Comment{\(\tsLoc_s\) is the initial location of current SCC}
                \STATE $\Gamma_s \gets \text{SCCs}[\text{cur}]$ \Comment{\emph{cur} is the index of current traversed SCC}
                \STATE $\mathbf{\eta_s} \gets \textit{InvProp}(\Gamma_s,DG(\Gamma_s),\tsLoc_s)$ 
                \Comment{Process single SCC}
                \FOR{\textbf{each} transition $\tau$ directed from $\tsLoc_s$ to $\tsLoc_t$} \label{line:for1}
                    \STATE $\text{next} \gets \textit{FindSCC}(\tsLoc_t,\text{SCCs})$
                    \IF{\(\text{next} \neq \text{cur} \)}
                       \STATE $\text{stack}.\textit{push}(\text{next},\tsLoc_t)$ \Comment{Traverse SCCs in BFS order}
                    \ENDIF
                \ENDFOR
                \STATE $\mathbf{\eta} \gets \textit{Merge}(\mathbf{\eta},\mathbf{\eta_s})$ \label{line:merge}
                \Comment{Combine assertion maps disjunctively}
            \ENDWHILE
            \RETURN $\mathbf{\eta}$ 
        \ENDIF \label{line:19}
        \STATE $\mathbf{\eta} \gets \textit{InitInv}(\Gamma,\tsLoc^*)$ \label{line:init_inv}
        \Comment{Compute invariant only in initial location}
        \STATE $\Gamma_s \gets \textit{Project}(\Gamma,\tsLoc^*)$ \label{line:project}
        \Comment{Derive sub-\LTS{} \(\Gamma_s\) by removing \(l^*\)}
        \FOR{\textbf{each} transition $\tau$ directed from $\tsLoc^*$ to $\tsLoc_t$} \label{line:for2}
        \STATE $\mathbf{\eta_s} \gets \textit{InvProp}(\Gamma_s,DG(\Gamma_s),\tsLoc_t)$
        \STATE $\mathbf{\eta} \gets \textit{Merge}(\mathbf{\eta},\mathbf{\eta_s})$ \label{line:merge2}
        \ENDFOR
        \RETURN $\mathbf{\eta}$ \label{line:return_res_single}
    \end{algorithmic}
    \label{alg:invprop}
\end{algorithm}

The algorithm consists of the following steps: First, we initialize the assertion map and use the classical Tarjan's algorithm~\cite{tarjan1972depth} to compute a list of strongly connected components (SCCs) in the directed graph \(DG(\Gamma)\) (lines~\ref{line:1}-\ref{line:2}). Depending on whether the graph is decomposable, i.e., whether the size of the SCC list is more than one, we consider two cases: 

(\roman{counter1}) For a directed graph that can be decomposed into multiple SCCs, we start from the entry SCC and traverse the list of SCCs in breadth-first order, computing invariants for each SCC recursively and integrating them into the final assertion map \(\eta\) (lines~\ref{line:3}-\ref{line:19}).

(\roman{counter2}) For a directed graph that is a single SCC, we compute the initial invariants at the starting location using the previous method~\cite{oopsla22/scalable}, then eliminate the start location \(\tsLoc^*\), traverse each edge originating from \( \tsLoc^*\) , propagate the invariants to the remaining sub-graph, and disjunctively merge the returned inductive assertion mappings to produce the final disjunctive invariant (lines~\ref{line:init_inv}-\ref{line:return_res_single}).

We formally define the specific functions involved in the algorithm as follows:

\begin{enumerate}
    \item \textbf{Merge}$(\eta_1, \eta_2)$ (line~\ref{line:merge} and line~\ref{line:merge2}). We extend $\eta$ to a mapping from the set of locations $\tsLocs$ to disjunction of affine assertions, specifically representing affine inequalities in DNF. The Merge function is thereby defined as a new mapping such that, for any location $\tsLoc \in \tsLocs$:
    \[
    \text{Merge}(\eta_1, \eta_2)(\tsLoc) = \eta_1(\tsLoc) \lor \eta_2(\tsLoc)
    \]
    
    \item \textbf{Project}$(\Gamma, \tsLoc^*)$ (line~\ref{line:project}). Considering the directed graph $\text{DG}(\Gamma)$ corresponding to the \LTS{} $\Gamma$, we remove all edges associated with the node $\tsLoc^*$, as well as the node itself. The derived \LTS{} corresponding to the resulting sub-graph is denoted by $\text{Project}(\Gamma, \tsLoc^*)$.
\end{enumerate}

Note that in Algorithm~\ref{alg:invprop}, we omitted the initial condition \(\tsInitcond\) and the propagation effect along the transitions. 
At each point of our algorithm (line~\ref{line:for1} and line~\ref{line:for2}) that tackles a transition \(\langle \tsLoc_s, \tsLoc_t,\tsGuardcond \rangle\), the propagation effect is computed as the post image of the conjunction of the invariant on \(\tsLoc_s\) and the guard \(\tsGuardcond\) via polyhedral projection onto the primed variables $\tsVars'$ and serves as the initial condition \(\tsInitcond\) of the new \LTS{} including \(\tsLoc_t\).

\begin{example}
Recall the example in Section~\ref{sec:overview}, specifically Figure~\ref{fig:ExampleofInvPropagation}. Here, \(\Gamma\) is an indivisible SCC. After computing the invariant $\eta(\tsLoc_2)$ of the ATS $\Gamma$ at the initial location $\tsLoc_2$, we consider all transitions (i.e. $\{\tau_4\}$) starting from the initial location $\tsLoc_2$, as depicted in the figure. Then, we propagate the invariant through the transition $\tau_4$ to \(\tsLoc_1\). After project to obtain the remaining sub-\LTS{} $\Gamma_{sub}$, composed of \(\tsLoc_1\) and its self-loop transition, we recursively compute this indivisible SCC to obtain the complete inductive assertion map. \qed
\end{example}

Our invariant propagation technique applies to all 
\LTS{}. The main advantage to incorporate this technique is that it allows the generation of invariants only at the initial locations of (sub-)SCCs, thus avoiding the generation of the invariants at all locations as adopted in~\cite{DBLP:conf/sas/SankaranarayananSM04,oopsla22/scalable}. In the case that the directed graph of the input \LTS{} is a cycle, our invariant propagation reaches the highest efficiency that generates the invariant only at the initial location of the cycle and derives invariants at other locations of the cycle by propagation, since the cycle has an explicit topological order after the removal of the initial location. This advantage becomes more prominent in loops with a non-neglectable amount of conditional branches.
The soundness of the invariant propagation is given in the following theorem.

\begin{theorem}
    \label{thm:propagation}
    The assertion maps generated by the invariant propagation algorithm are inductive.
\end{theorem}

\begin{proof}
We prove by induction on the number $k$ of locations in the input \LTS{} \(\Gamma\)  that the assertion map obtained by our invariant propagation algorithm for the \LTS{} \(\Gamma\) is inductive. 

We first consider the base case, i.e., \(k = 1\). In this case, \(DG(\Gamma)\) has only one location, which is obviously indivisible. Here, the function \textbf{InitInv}(), previously mentioned as applying Farkas' Lemma for conjunctive invariant computation, is called. Therefore, the resulting assertion map is inductive, and its correctness is guaranteed by the prior results.

Assuming that the case when the size of $\Gamma$ equals \(k\) holds, we prove that it holds for $\Gamma$ of size \(k + 1\). 
For an \LTS{} \(\Gamma\) with \(k + 1\) locations, if it is divisible, it can be decomposed into several sub-SCCs \(\Gamma_{sub}\) with sizes less than or equal to \(k\). After the call to function \(\textbf{InvProp}()\) at Line 8, we obtain an inductive assertion mapping by the inductive condition. The \(\textbf{Merge}\)() function does not affect the inductive condition of the combined mapping. On the other hand, if it is indivisible, then our approach computes the invariant at its initial location and, after projecting away the initial location \(\tsLoc^*\), obtains a sub-\LTS{} \(\Gamma_{sub}\) of size \(k\). Similarly, the recursive call to invariant propagation at Line 20 and merging the returned results always yields an inductive assertion map by the inductive condition. \qed
\end{proof}

\subsection{Other Optimizations}\label{sec:otheropt}

\smalltitle{Loop Summary.} To address more general control flow, such as nested loops, we use the standard method of loop summary to express the input-output relationship of the inner loops (while adding fresh variables for input values) to handle nested loops. See Appendix~\ref{app:loop_sum} for the technical details.

\smalltitle{Infeasible Implication.} In the previous results~\cite{DBLP:conf/sas/SankaranarayananSM04,oopsla22/scalable}, the infeasible implication is not handled in their prototype. 
Recall the infeasible implication corresponding to $\eta(\tsLoc) \land \rho \models -1 \geq 0$ illustrated in Figure~\ref{tab:farkascons}. To fully address this issue, we can simply set \(\mu=1\) in Figure~\ref{tab:farkascons} so that the nonlinear multiplier \(\mu\) is eliminated. The correctness is given by the following theorem. 

\begin{theorem} \label{thr:infeasible}
Let \(\Gamma\) be an ATS. For any AAM \(\eta\) that fulfills the initial and consecution conditions derived from the ATS \(\Gamma\) with the original constraints for the infeasible implication as in each consecution tabular of Figure~\ref{tab:farkascons} (aimed at \(-1>=0\)) with each $\mu$ in an infeasible implication instantiated as $k$ for some $k>0$, it is equivalent to setting all $\mu$'s to $1$ while preserving the constraints of infeasible implication.
\end{theorem}

We present our proof in Appendix~\ref{app:thm_infeasible} and Appendix~\ref{sec:appendix_minkowski}. The main idea of the proof is that, for the infeasible implication case,  by scaling each \(\lambda_i\) other than \(\mu\), the consecution tabular used to generate polyhedra is transformed into an equivalent tabular with scaled lambda variables \(\lambda'_i\). so that it suffices to choose the multiplier $\mu$ to be $1$. 

\begin{remark}[Extensions]\label{rmk:extension}
Our approach can be extended in the following ways. To obtain a more precise path condition representation, one extension is by (i) distinguishing the remainders (e.g., even/odd) modulo a fixed positive integer (e.g. 2) when handling modular arithmetic and (ii) detecting hidden termination phases via the approach in \cite{DBLP:conf/cav/Ben-AmramG17}. To handle machine integers, another extension is by having a piecewise disjunctive treatment for the cases of overflow and non-overflow. 
By applying standard bottom-up analysis for interprocedural programs, our approach can also handle programs with function calls.
Finally, our approach could be further extended to floating point numbers by considering piecewise affine approximations~\cite{DBLP:conf/esop/Mine04,DBLP:conf/vmcai/Mine06}. 
\end{remark}

\section{Experimental Evaluation}
\label{sec:exp}
In this section, we present the evaluation of the implementation (referred to as {\ToolName}) of our approach to generate disjunctive affine invariants. We focus on the following two questions (\textbf{RQ1} and \textbf{RQ2}).

\begin{itemize}
    \item \textbf{RQ1:} How competitive is \ToolName\ when compared with other approaches?
    \item \textbf{RQ2:} How effective does invariant propagation enhance our approach?
\end{itemize}
\subsection{Experimental Setup}

\smalltitle{Implementation.} We implement our approach (including the algorithmic techniques in Section~\ref{sec:alg}) as a prototype \ToolName, dividing the implementation into front-end and back-end. The front-end utilizes Clang Static Analyzer~\cite{ClangStaticAnalyzer} to extract and transform C programs, processing programs into the format required by the back-end. The back-end is an extension of StInG~\cite{Sting} written in C++ and uses PPL 1.2~\cite{DBLP:conf/sas/BagnaraRZH02} for polyhedra manipulation (e.g., projection, generator computation, etc.), which generates invariants and propagate them to obtain a disjuntive invariant as the loop invariant.

\smalltitle{Environment.} All experiments are conducted on a machine equipped with a 12th-generation Intel(R) Core(TM) i7-12800HX CPU, 16 cores, 2304 MHz, 9.5GB RAM, running Ubuntu 20.04 (LTS). Following the competition settings of SV-COMP, for studies \textbf{RQ1} and \textbf{RQ2}, we impose a time limit of 900s.

\smalltitle{Benchmarks.} We have a total of 114 affine programs, $38.6\%$ of which have disjunctive features,
sourced from: 1) 105 benchmarks from the SV-COMP, ReachSafety-Loop track. We excluded those with arrays, pointers, and other non-numeric features, those with modulus, division, polynomial, and other non-linear operations. 2) 9 benchmarks from the recent paper~\cite{DBLP:conf/vmcai/BoutonnetH19}, which include complex nested loops and examples with disjunctive features.

\smalltitle{Methodology.} 
In \textbf{RQ1}, we compare \ToolName\ utilizing invariant propagation techniques with several state-of-the-art software verifiers:

\begin{itemize}
    \item Veriabs~\cite{SVCOMP2023Veriabs} is a state-of-the-art software verifier that is an integration of various strategies such as fuzz testing, $k$-induction, loop shrinking, loop pruning, full-program induction, explicit state model checking and other invariant generation techniques, which is capable to deal with programs with disjunctive features.
    \item CPAChecker~\cite{CPAchecker} is a well-developed software verifier that is based on bounded model checking and interpolation and has a comprehensive ability to verify various kinds of properties. 
    \item OOPSLA23~\cite{oopsla23} is a recent recurrence analysis tool that handles only loops with the ultimate strict alternation pattern that eventually the loop will alternate between different modes periodically and performs good on such class of programs, which thus excels in the verification of disjunctive programs with alternating modes.
    \item DIG~\cite{ICSE2022dig} is an invariant generation tool considering disjunctive features in programs and utilizes front-end CIVL~\cite{DIG_CIVL} to obtain symbolic execution traces. It employs dynamic analysis along with efficient algorithms from algebra and geometry to solve numerical invariant templates, thereby generating numerical invariants at any position within a program, which is capable of extensively handling the programs with array, nonlinear, linear and disjunctive features.
    \item IKOS with \emph{Polyset} domain from PPLite~\cite{IKOS,PPLite_domain} is a classic abstract interpretation framework with various interface supports. The \emph{Polyset} abstract domain is an efficient implementation of the powerset of polyhedra and serves as an alternative to the trace partitioning strategy implemented in Astree~\cite{ASTree}.
\end{itemize}

In \textbf{RQ2}, we focus on comparing the impact of the invariant propagation technique on the time efficiency. By contrasting the tool's performance when calculating invariants for each location individually against using invariant propagation, we analyze the role of invariant propagation.

\subsection{Tool Comparison (RQ1)}

Our work primarily focuses on the generation of disjunctive invariants, whereas tools like CPAChecker and Veriabs are specifically designed as bug finders for verifying assertions. However, by integrating the PPL library~\cite{DBLP:conf/sas/BagnaraRZH02} and Z3~\cite{z3}, we use the generated invariants to verify the correctness of assertions and demonstrate the precision of the invariants generated by \ToolName.

\begin{table}[t]
\centering
\small
\resizebox{0.95\textwidth}{!}{
\begin{tabular}{|cc|ccc|ccc|ccc|}
\hline
\multicolumn{2}{|c|}{Benchmark}                                                    & \multicolumn{3}{c|}{\ToolName }                                               & \multicolumn{3}{c|}{Veriabs}                                         & \multicolumn{3}{c|}{CPAChecker}                                                           \\ \hline
\multicolumn{1}{|c|}{Source}                       & \#Num                         & \multicolumn{1}{c|}{\#Ver.} & \multicolumn{1}{c|}{\#Unk.} & Time (s)                        & \multicolumn{1}{c|}{\#Ver.} & \multicolumn{1}{c|}{\#Unk.} & Time (s) & \multicolumn{1}{c|}{\#Ver.} & \multicolumn{1}{c|}{\#Unk.} & Time (s)                      \\ \hline
\multicolumn{1}{|c|}{loop-invariants}              & 5                             & \multicolumn{1}{c|}{4}      & \multicolumn{1}{c|}{1}      & 0.47                            & \multicolumn{1}{c|}{5}      & \multicolumn{1}{c|}{0}      & 153.31   & \multicolumn{1}{c|}{4}      & \multicolumn{1}{c|}{1}      & 1001.49                       \\ \hline
\multicolumn{1}{|c|}{loop-new}                     & 2                             & \multicolumn{1}{c|}{2}      & \multicolumn{1}{c|}{0}      & 0.11                            & \multicolumn{1}{c|}{0}      & \multicolumn{1}{c|}{2}      & 959.74   & \multicolumn{1}{c|}{0}      & \multicolumn{1}{c|}{2}      & 1807.20                       \\ \hline
\multicolumn{1}{|c|}{loop-invgen}                  & 5                             & \multicolumn{1}{c|}{4}      & \multicolumn{1}{c|}{1}      & 0.41                            & \multicolumn{1}{c|}{5}      & \multicolumn{1}{c|}{0}      & 160.51   & \multicolumn{1}{c|}{0}      & \multicolumn{1}{c|}{5}      & 4518.19                       \\ \hline
\multicolumn{1}{|c|}{loops-crafted-1}              & 25                            & \multicolumn{1}{c|}{20}     & \multicolumn{1}{c|}{5}      & 4.84                            & \multicolumn{1}{c|}{25}     & \multicolumn{1}{c|}{0}      & 4010.55  & \multicolumn{1}{c|}{0}      & \multicolumn{1}{c|}{25}     & 22607.55                      \\ \hline
\multicolumn{1}{|c|}{loop-simple}                  & 2                             & \multicolumn{1}{c|}{1}      & \multicolumn{1}{c|}{1}      & 2                               & \multicolumn{1}{c|}{1}      & \multicolumn{1}{c|}{1}      & 944.69   & \multicolumn{1}{c|}{1}      & \multicolumn{1}{c|}{1}      & 919.03                        \\ \hline
\multicolumn{1}{|c|}{loop-zilu}                    & 26                            & \multicolumn{1}{c|}{26}     & \multicolumn{1}{c|}{0}      & 0.77                            & \multicolumn{1}{c|}{25}     & \multicolumn{1}{c|}{1}      & 1064.60  & \multicolumn{1}{c|}{26}     & \multicolumn{1}{c|}{0}      & 307.18                        \\ \hline
\multicolumn{1}{|c|}{loops}                        & 18                            & \multicolumn{1}{c|}{15}     & \multicolumn{1}{c|}{3}      & 3.26                            & \multicolumn{1}{c|}{17}     & \multicolumn{1}{c|}{1}      & 536.33   & \multicolumn{1}{c|}{17}     & \multicolumn{1}{c|}{1}      & 1123.35                       \\ \hline
\multicolumn{1}{|c|}{loop-lit}                     & 10                            & \multicolumn{1}{c|}{10}     & \multicolumn{1}{c|}{0}      & 22.22                           & \multicolumn{1}{c|}{10}     & \multicolumn{1}{c|}{0}      & 280.87   & \multicolumn{1}{c|}{5}      & \multicolumn{1}{c|}{5}      & 5655.33                       \\ \hline
\multicolumn{1}{|c|}{loop-acceleration}            & 10                            & \multicolumn{1}{c|}{9}      & \multicolumn{1}{c|}{1}      & 0.32                            & \multicolumn{1}{c|}{9}      & \multicolumn{1}{c|}{1}      & 493.13   & \multicolumn{1}{c|}{9}      & \multicolumn{1}{c|}{1}      & 1030.78                       \\ \hline
\multicolumn{1}{|c|}{loop-crafted}                 & 2                             & \multicolumn{1}{c|}{2}      & \multicolumn{1}{c|}{0}      & 0.09                            & \multicolumn{1}{c|}{2}      & \multicolumn{1}{c|}{0}      & 49.59    & \multicolumn{1}{c|}{2}      & \multicolumn{1}{c|}{0}      & 27.97                         \\ \hline
\multicolumn{1}{|c|}{\cite{DBLP:conf/vmcai/BoutonnetH19}} & 9                             & \multicolumn{1}{c|}{8}      & \multicolumn{1}{c|}{1}      & 2.12                            & \multicolumn{1}{c|}{9}      & \multicolumn{1}{c|}{0}      & 286.24   & \multicolumn{1}{c|}{4}      & \multicolumn{1}{c|}{5}      & 4576.98                       \\ \hline
\multicolumn{1}{|c|}{\textbf{Total}}                        & \textbf{114} & \multicolumn{1}{c|}{101}    & \multicolumn{1}{c|}{13}     & \textbf{34.65} & \multicolumn{1}{c|}{108}    & \multicolumn{1}{c|}{6}      & 8939.56  & \multicolumn{1}{c|}{68}     & \multicolumn{1}{c|}{46}     & 43573.42                      \\ \hline \hline
\multicolumn{2}{|c|}{Benchmark}                                                    & \multicolumn{3}{c|}{OOPSLA23}                                                               & \multicolumn{3}{c|}{DIG}                                             & \multicolumn{3}{c|}{IKOS + PPLite}                                                        \\ \hline
\multicolumn{1}{|c|}{Source}                       & \#Num                         & \multicolumn{1}{c|}{\#Ver.} & \multicolumn{1}{c|}{\#Unk.} & Time (s)                        & \multicolumn{1}{c|}{\#Ver.} & \multicolumn{1}{c|}{\#Unk.} & Time (s) & \multicolumn{1}{c|}{\#Ver.} & \multicolumn{1}{c|}{\#Unk.} & \multicolumn{1}{c|}{Time (s)} \\ \hline
\multicolumn{1}{|c|}{loop-invariants}              & 5                             & \multicolumn{1}{c|}{1}      & \multicolumn{1}{c|}{4}      & 14.09                           & \multicolumn{1}{c|}{0}       & \multicolumn{1}{c|}{5}       & 2344.12          & \multicolumn{1}{c|}{3}       & \multicolumn{1}{c|}{2}       & 0.88       \\ \hline
\multicolumn{1}{|c|}{loop-new}                     & 2                             & \multicolumn{1}{c|}{0}      & \multicolumn{1}{c|}{2}      & 5.83                            & \multicolumn{1}{c|}{0}       & \multicolumn{1}{c|}{2}       & 241.68         & \multicolumn{1}{c|}{1}       & \multicolumn{1}{c|}{1}       & 988.06        \\ \hline
\multicolumn{1}{|c|}{loop-invgen}                  & 5                             & \multicolumn{1}{c|}{4}      & \multicolumn{1}{c|}{1}      & 14.78                           & \multicolumn{1}{c|}{1}       & \multicolumn{1}{c|}{4}       & 264.34         & \multicolumn{1}{c|}{5}       & \multicolumn{1}{c|}{0}       & 1.03        \\ \hline
\multicolumn{1}{|c|}{loops-crafted-1}              & 25                            & \multicolumn{1}{c|}{22}     & \multicolumn{1}{c|}{3}      & 82.03                           & \multicolumn{1}{c|}{10}       & \multicolumn{1}{c|}{15}       & 6030.28         & \multicolumn{1}{c|}{0}       & \multicolumn{1}{c|}{25}       & 8270.92        \\ \hline
\multicolumn{1}{|c|}{loop-simple}                  & 2                             & \multicolumn{1}{c|}{0}      & \multicolumn{1}{c|}{2}      & 5.82                            & \multicolumn{1}{c|}{0}       & \multicolumn{1}{c|}{2}       &  493.24        & \multicolumn{1}{c|}{2}       & \multicolumn{1}{c|}{0}       & 10.01       \\ \hline
\multicolumn{1}{|c|}{loop-zilu}                    & 26                            & \multicolumn{1}{c|}{0}      & \multicolumn{1}{c|}{26}     & 68.24                           & \multicolumn{1}{c|}{19}       & \multicolumn{1}{c|}{7}       & 6878.38         & \multicolumn{1}{c|}{17}       & \multicolumn{1}{c|}{9}       & 1827.31       \\ \hline
\multicolumn{1}{|c|}{loops}                        & 18                            & \multicolumn{1}{c|}{4}      & \multicolumn{1}{c|}{14}     & 48.36                           & \multicolumn{1}{c|}{3}       & \multicolumn{1}{c|}{15}       & 5241.99         & \multicolumn{1}{c|}{7}       & \multicolumn{1}{c|}{11}       & 909.04      \\ \hline
\multicolumn{1}{|c|}{loop-lit}                     & 10                            & \multicolumn{1}{c|}{7}      & \multicolumn{1}{c|}{3}      & 29.59                           & \multicolumn{1}{c|}{3}       & \multicolumn{1}{c|}{7}       & 2024.29         & \multicolumn{1}{c|}{5}       & \multicolumn{1}{c|}{5}       & 3993.86     \\ \hline
\multicolumn{1}{|c|}{loop-acceleration}            & 10                            & \multicolumn{1}{c|}{8}      & \multicolumn{1}{c|}{1}      & 27.48                           & \multicolumn{1}{c|}{3}       & \multicolumn{1}{c|}{7}       & 2263.40         & \multicolumn{1}{c|}{6}       & \multicolumn{1}{c|}{4}       & 1.66       \\ \hline
\multicolumn{1}{|c|}{loop-crafted}                 & 2                             & \multicolumn{1}{c|}{2}      & \multicolumn{1}{c|}{0}      & 5.63                            & \multicolumn{1}{c|}{0}       & \multicolumn{1}{c|}{2}       &  503.76        & \multicolumn{1}{c|}{2}       & \multicolumn{1}{c|}{0}       & 0.33      \\ \hline
\multicolumn{1}{|c|}{\cite{DBLP:conf/vmcai/BoutonnetH19}} & 9                             & \multicolumn{1}{c|}{6}      & \multicolumn{1}{c|}{3}      & 28.98                           & \multicolumn{1}{c|}{1}       & \multicolumn{1}{c|}{8}       & 497.92         & \multicolumn{1}{c|}{8}       & \multicolumn{1}{c|}{1}       & 2.04        \\ \hline
\multicolumn{1}{|c|}{\textbf{Total}}                        & \textbf{114} & \multicolumn{1}{c|}{55}     & \multicolumn{1}{c|}{59}     & 330.83                          & \multicolumn{1}{c|}{40}       & \multicolumn{1}{c|}{74}       & 26783.40         & \multicolumn{1}{c|}{56}       & \multicolumn{1}{c|}{58}       & 16005.15        \\ \hline
\end{tabular}
}
\vspace{0.8em}
\caption{Comparisons Over 114 Benchmarks}

\label{exp:ComparisonOverTools}
\end{table}

The complete comparison results of \ToolName\ with other tools are presented in Table~\ref{exp:ComparisonOverTools}. In the table, \textit{Source} indicates the source category of the benchmark. The term \textit{\#Ver.} represents the number of examples correctly verified by the verifier, and \textit{\#Unk.} (unknown) mainly arises from the following situations: a) The front-end fails to parse correctly, resulting in program crashes. b) Returns \textbf{Unknown}. c) Timeouts. For the benchmarks from~\cite{DBLP:conf/vmcai/BoutonnetH19}, which do not contain assertions to be verified, we modify the invariants generated by our \ToolName\ as assertions and test them over the other tools to obtain results.

From the table, it is evident that \ToolName\ typically requires less than 0.3 seconds on average for verification, and its overall verification accuracy is very close to that of the SV-COMP 2023 Reachability track winner Veriabs, while significantly outperforming Veriabs in terms of time efficiency by 10X to 1000X. This is mainly because Veriabs employs a rich strategy to assist verification, granting it a stronger verification capability but also requiring more time for most examples. CPAChecker experienced a broad range of timeouts in examples with complex loops that could not be verified within a finite unfolding of loops. This is due to the intrinsic limitations of its bounded model checking approach, and its loop unwinding strategy also results in verification times on the dataset that significantly exceed those of other tools. 

Despite the fact that the tool from~\cite{oopsla23} has the second fewest number of verified benchmarks, it outperforms other tools in examples suitable for recurrence analysis. For DIG, we employ it to generate loop invariants and post conditions, and use Z3~\cite{z3} prover to verify the assertion. Nevertheless, the frontend of DIG necessitates CIVL's reliance on extracting symbolic execution paths from the program. When processing loops, it similarly depends on loop unrolling, and if it cannot fully unroll loops within a small bound, it determines that locations after the loop are unreachable. Consequently, it exhibits issues analogous to those of CPAChecker. Additionally, for some randomly assigned variables in SV-COMP, DIG lacks a suitable modeling. We have already reported several bugs via issues on GitHub. As a classical framework for abstract interpretation, IKOS with PPLite did not deliver optimal verification outcomes on the dataset. In some straightforward nested loops and more extensive loop iterations, it either failed to converge to a fixed point, or the precision of the invariants obtained upon convergence was insufficient to verify assertions, thereby causing timeouts or unknown in certain instances. 

In summary, we conclude that \ToolName\ significantly outperforms other tools such as Veriabs in time efficiency for affine numerical programs, while its verification capability is not inferior to the SV-COMP winner Veriabs. We also conducted an in-depth analysis of the cases where our \ToolName\ returns \textbf{Unknown}. The primary reasons for the issues include: a) the absence of type range constraints at the front end, b) reliance on modular arithmetic, c) the need for more complex loop generalizations, d) exceeding the computational precision of the PPL library, and e) exponential arithmetic that surpasses the modeling capabilities of linear templates. 7-8 of these cases could be further solved by optimizing implementations. In the verifiable cases, the preliminary implementation of \ToolName\ has already far surpassed existing methods in efficiency.

\subsection{Ablation Study in Invariant Propagation (RQ2)}

\begin{table}[t]
    \centering
    \resizebox{0.9\linewidth}{!}{%
        \begin{tabular}{|cc|cccccc|}
        \hline
        \multicolumn{2}{|c|}{\multirow{2}{*}{Benchmark}} & \multicolumn{6}{c|}{\ToolName}                                                                                                                                    \\ \cline{3-8} 
        \multicolumn{2}{|c|}{}                           & \multicolumn{3}{c|}{No PPG}                                                               & \multicolumn{3}{c|}{PPG}                                             \\ \hline
        \multicolumn{1}{|c|}{Source}        & \#Num      & \multicolumn{1}{c|}{\#Ver.} & \multicolumn{1}{c|}{\#Unk.} & \multicolumn{1}{c|}{Time (s)} & \multicolumn{1}{c|}{\#Ver.} & \multicolumn{1}{c|}{\#Unk.} & Time (s) \\ \hline
        \multicolumn{1}{|c|}{SV-COMP}       & 105        & \multicolumn{1}{c|}{91}     & \multicolumn{1}{c|}{14}     & \multicolumn{1}{c|}{1825.53}  & \multicolumn{1}{c|}{93}     & \multicolumn{1}{c|}{12}     & 32.53    \\ \hline
        \multicolumn{1}{|c|}{paper}         & 9          & \multicolumn{1}{c|}{8}      & \multicolumn{1}{c|}{1}      & \multicolumn{1}{c|}{10.76}    & \multicolumn{1}{c|}{8}      & \multicolumn{1}{c|}{1}      & 2.12     \\ \hline
        \end{tabular}
    }
    \vspace{0.8em}
    \caption{Experiment for Invariant Propagation}
    \label{tab:propagation}
\end{table}
\begin{figure}[t]
    \centering
    \includegraphics[width=0.6\linewidth,trim={1cm 0.5cm 0 0.5cm}]{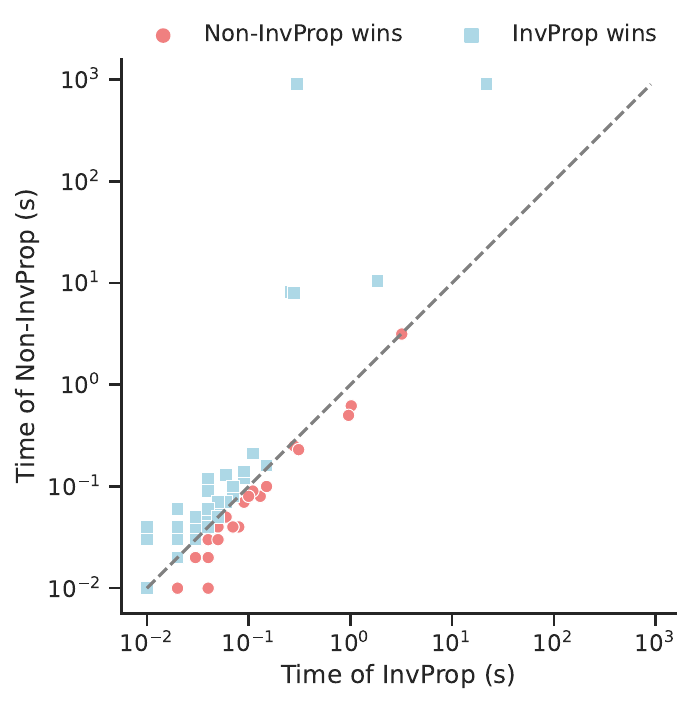}
    \caption{Comparison for Invariant Propagation}
    \label{fig:CompareProgation}
\end{figure}

In this section, we conduct an ablation study to evaluate the performance of the invariant propagation technique within \ToolName. In Table~\ref{tab:propagation}, we present the overall results, where we can clearly observe that the use of invariant propagation leads to a 5X-50X improvement in time efficiency.

More specifically, through the scatter plot in Figure~\ref{fig:CompareProgation}, we compared the time performance of individual examples before and after the application of invariant propagation techniques. In some cases, invariant propagation led to significant efficiency improvements (10X-1000X). This is due to the fact that for more complex programs, the size of the \LTS{} \(\Gamma\) is larger, and applying invariant propagation techniques on this basis can maximize performance optimization. Since the tool itself performs efficiently in most examples, the optimization brought by this technique is not apparent in those cases in the graph where the time is below 0.1 seconds. As the propagation itself, including the projection of sub-ATS, incurs a certain time cost, which dilutes the time optimization brought about by invariant propagation.

In conclusion, invariant propagation significantly enhances the tool's scalability and yields superior optimization results for complex examples. This also reveals that, within our constraint-solving methodology, the cost of computing invariants at any given location constitutes the principal computational bottleneck. By reducing the number of locations that need to be computed and leveraging prior results to avoid redundant polyhedral operations, we can effectively enhance efficiency.

\subsection{Caveat to Correctness}

This section elucidates configurations that may induce subtle deviations from real-world programs or alternative models during the empirical evaluation of our tool.

\begin{itemize}
    \item In our current experimental setup, we have not accounted for the behavior of machine integers during overflow conditions. Consequently, our verification process is confined to affine programs that do not encounter overflow errors.
    \item Within the context of control flow transformations, we introduce uncertainty into conditional statements by adhering to the SV-COMP guidelines. This is achieved by replacing branch conditions with functions that return random Boolean values, thereby emulating the semantics of non-deterministic branches. Nonetheless, we have yet to effectively model uncertainty in variable coefficients, specifically affine inequalities with coefficients represented as intervals.
\end{itemize}

\section{Related Works}
\label{sec:related_work}
Our methodology enhances conjunctive affine invariants by integrating optimizations from prior research~\cite{DBLP:conf/cav/ColonSS03,DBLP:conf/sas/SankaranarayananSM04,oopsla22/scalable,DBLP:conf/cav/JiFFC22} and utilizing control flow transformation techniques to extend them to disjunctive forms. A principal contribution of this paper is the mitigation of computational inefficiencies resulting from the exponential state space expansion associated with disjunctive extensions, achieved through invariant propagation. Consequently, this approach distinguishes our work from existing studies.
The work~\cite{DBLP:conf/pldi/GulwaniSV08} generates disjunctive invariants by predefining disjunctive templates, heuristically selecting physical cut points (while we select abstract locations from loop paths) and transforming the quadratic constraints from Farkas' Lemma into SAT solving.
Other approaches for conjunctive affine invariant generation include ~\cite{DBLP:conf/atva/OliveiraBP17,DBLP:conf/cav/GuptaR09}.
These approaches propose completely different techniques, and thus are orthogonal to our approach.

Polynomial invariant generation~\cite{DBLP:conf/dagstuhl/Kapur05,DBLP:journals/fcsc/YangZZX10,DBLP:conf/pldi/Chatterjee0GG20,DBLP:conf/lics/HrushovskiOP018,DBLP:conf/issac/Rodriguez-CarbonellK04,DBLP:conf/vmcai/Cousot05,DBLP:conf/sas/AdjeGM15,DBLP:journals/fcsc/LinWYZ14,DBLP:conf/cav/ChenHWZ15,DBLP:conf/atva/OliveiraBP16,DBLP:conf/issac/HumenbergerJK17,DBLP:conf/popl/SankaranarayananSM04} has been widely investigated.
Most of these approaches consider conjunctive polynomial invariants only. 
Compared with conjunctive polynomial invariants, disjunctive affine invariants capture the precise feature of phase and mode changes in affine loops, and therefore are more precise. 

The works~\cite{DBLP:conf/sigsoft/XieCLLL16,DBLP:conf/tase/LinZCSXLS21} are based on path dependency automata, requiring precise estimates of the number of iterations in loops, which limits their analysis to programs with regular alternation and inductive variables (computable general terms). The work~\cite{DBLP:conf/cav/SharmaDDA11} studies the detection of multiphase disjunctive invariants. Multiphase invariants are a special case of our control flow transformation since each phase in a multiphase loop cannot go back to previous phases, while in our control flow transformation, locations can go back and forth via transitions. Thus, we have a wider class of disjunctive invariants as compared with~\cite{DBLP:conf/cav/SharmaDDA11}. 

Our control flow transformation is related to control flow refinement~\cite{DBLP:conf/emsoft/BalakrishnanSIG09,DBLP:conf/pldi/GulwaniJK09,DBLP:journals/pacmpl/CyphertBKR19,DBLP:conf/cav/SilvermanK19} in the literature. 
These approaches mostly focus on representing the control flow of multiple loop iterations as regular expressions and refine these regular expressions by various approaches such as abstract domains, simulation relation and even invariant generation to reduce infeasible paths. 
Our control flow transformation considers the loop body within a single loop iteration, and is dedicated to the application of Farkas' Lemma. Thus, our control flow transformation has a different focus compared with these results. 
Moreover, the use of Farkas' Lemma can circumvent the issue that finer control flow may not always lead to finer analysis in control flow refinement~\cite{DBLP:journals/pacmpl/CyphertBKR19}. 

Our invariant propagation is related to abstract interpretation~\cite{DBLP:conf/popl/CousotH78,DBLP:conf/sas/BagnaraHRZ03,DBLP:conf/popl/SinghPV17,DBLP:conf/sas/GopanR07,DBLP:conf/vmcai/BoutonnetH19,DBLP:journals/entcs/HenryMM12}. The main difference is that it  propagates the \emph{already-computed} invariants (via Farkas' Lemma) to yet not computed locations as much as possible to minimize the invariant generation computation, while abstract interpretation usually requires an involved fixed-point iteration to \emph{compute} invariants.

Recurrence analysis~\cite{DBLP:conf/fmcad/FarzanK15,DBLP:conf/pldi/KincaidBBR17,DBLP:journals/pacmpl/KincaidCBR18} works well over programs with specific structure that ensures closed form solutions. 
For example, the most related recurrence analysis approach~\cite{oopsla23} (that also targets disjunctive invariants) solves the exact invariant over the class of loops with (ultimate) strict alternation between different modes.  
Compared with recurrence analysis, our approach does not require specific program structure to ensure closed form solution, but is less precise over programs that can be solved exactly by recurrence analysis.

Finally, we compare our approach with other methods such as machine learning, inference and data-driven approaches. 
Unlike constraint solving that can have an accuracy guarantee for the generated invariants based on the constraints, these methods cannot have an accuracy guarantee. 
Furthermore, machine learning and data-driven approaches themselves cannot guarantee that the generated assertions are indeed invariants. 
Moreover, our approach can generate invariants \emph{without} the need of a goal property, while several approaches (such as IC3~\cite{DBLP:conf/fmcad/SomenziB11}, CLN2INV~\cite{DBLP:conf/iclr/RyanWYGJ20}, \cite{FSE2022}) usually requires a goal property.
Note that the invariant generation without a given goal property is a classical setting (see e.g. ~\cite{DBLP:conf/cav/ColonSS03,DBLP:conf/popl/CousotH78}), and has applications in loop summary and probabilistic program verification (see e.g.~\cite{DBLP:conf/cav/ChakarovS13,DBLP:conf/pldi/WangS0CG21}). 

LLM-based invariant generation methods~\cite{AutoSpec} performs poorly on certain complex programs exhibiting disjunctive features. Those large-scale models have been unable to precisely comprehend the disjunctive properties inherent in these programs, and the invariants they produce often necessitate iterative interaction with Frama-C until an invariant that can be successfully verified by Frama-C is generated.
\section{Conclusion and Future Work}\label{sec:conclusion}

In this work, we propose a novel approach to generate affine disjunctive invariants over affine loops. Our novelty lies in combining a control flow transformation to extract the interleaving relationships between loop paths and employing Farkas' Lemma to solve the disjunctive invariants of loops. Additionally, we apply invariant propagation techniques to mitigate the computational costs of exponential explosion. A thorough resolution of the infeasible implication in the application of Farkas' Lemma and an extension to nested loops through loop summary are proposed as optimizations for practical program verification.
Experimental results show that our approach is competitive with state-of-the-art software verifiers in affine disjunctive invariant generation over affine loops. 
One future direction would be to consider extensions mentioned in Remark~\ref{rmk:extension}.

\begin{credits}
\subsubsection{\ackname} We thank anonymous reviewers for constructive comments. This work is partially supported by the National Natural Science Foundation of China (NSFC) under Grant No. 61872232 and No. 62172271.
\end{credits}

\clearpage
%% Bibliography
\bibliographystyle{splncs04}
\bibliography{invariants}

\begin{thebibliography}{10}
\providecommand{\url}[1]{\texttt{#1}}
\providecommand{\urlprefix}{URL }
\providecommand{\doi}[1]{https://doi.org/#1}

\bibitem{DBLP:conf/sas/AdjeGM15}
Adj{\'{e}}, A., Garoche, P., Magron, V.: Property-based polynomial invariant generation using sums-of-squares optimization. In: {SAS}. LNCS, vol.~9291, pp. 235--251. Springer, [S.l.] (2015)

\bibitem{DBLP:conf/cav/AlbarghouthiLGC12}
Albarghouthi, A., Li, Y., Gurfinkel, A., Chechik, M.: Ufo: {A} framework for abstraction- and interpolation-based software verification. In: {CAV}. LNCS, vol.~7358, pp. 672--678. Springer (2012). \doi{10.1007/978-3-642-31424-7\_48}, \url{https://doi.org/10.1007/978-3-642-31424-7\_48}

\bibitem{DBLP:conf/sas/AliasDFG10}
Alias, C., Darte, A., Feautrier, P., Gonnord, L.: Multi-dimensional rankings, program termination, and complexity bounds of flowchart programs. In: {SAS}. LNCS, vol.~6337, pp. 117--133. Springer (2010). \doi{10.1007/978-3-642-15769-1\_8}, \url{https://doi.org/10.1007/978-3-642-15769-1\_8}

\bibitem{DBLP:conf/pldi/AsadiC0GM21}
Asadi, A., Chatterjee, K., Fu, H., Goharshady, A.K., Mahdavi, M.: Polynomial reachability witnesses via stellens{\"{a}}tze. In: {PLDI}. pp. 772--787. {ACM} (2021). \doi{10.1145/3453483.3454076}, \url{https://doi.org/10.1145/3453483.3454076}

\bibitem{DBLP:conf/sas/BagnaraHRZ03}
Bagnara, R., Hill, P.M., Ricci, E., Zaffanella, E.: Precise widening operators for convex polyhedra. In: Cousot, R. (ed.) Static Analysis, 10th International Symposium, {SAS} 2003, San Diego, CA, USA, June 11-13, 2003, Proceedings. Lecture Notes in Computer Science, vol.~2694, pp. 337--354. Springer (2003). \doi{10.1007/3-540-44898-5\_19}, \url{https://doi.org/10.1007/3-540-44898-5\_19}

\bibitem{DBLP:conf/sas/BagnaraRZH02}
Bagnara, R., Ricci, E., Zaffanella, E., Hill, P.M.: Possibly not closed convex polyhedra and the parma polyhedra library. In: {SAS}. Lecture Notes in Computer Science, vol.~2477, pp. 213--229. Springer (2002). \doi{10.1007/3-540-45789-5\_17}, \url{https://doi.org/10.1007/3-540-45789-5\_17}

\bibitem{DBLP:conf/emsoft/BalakrishnanSIG09}
Balakrishnan, G., Sankaranarayanan, S., Ivancic, F., Gupta, A.: Refining the control structure of loops using static analysis. In: Chakraborty, S., Halbwachs, N. (eds.) Proceedings of the 9th {ACM} {\&} {IEEE} International conference on Embedded software, {EMSOFT} 2009, Grenoble, France, October 12-16, 2009. pp. 49--58. {ACM} (2009). \doi{10.1145/1629335.1629343}, \url{https://doi.org/10.1145/1629335.1629343}

\bibitem{PPLite_domain}
Becchi, A., Zaffanella, E.: Pplite: zero-overhead encoding of nnc polyhedra. Information and Computation  \textbf{275},  104620 (2020)

\bibitem{DBLP:conf/cav/Ben-AmramG17}
Ben{-}Amram, A.M., Genaim, S.: On multiphase-linear ranking functions. In: Majumdar, R., Kuncak, V. (eds.) {CAV}. LNCS, vol. 10427, pp. 601--620. Springer (2017). \doi{10.1007/978-3-319-63390-9\_32}, \url{https://doi.org/10.1007/978-3-319-63390-9\_32}

\bibitem{DBLP:conf/vmcai/BoutonnetH19}
Boutonnet, R., Halbwachs, N.: Disjunctive relational abstract interpretation for interprocedural program analysis. In: Enea, C., Piskac, R. (eds.) Verification, Model Checking, and Abstract Interpretation - 20th International Conference, {VMCAI} 2019, Cascais, Portugal, January 13-15, 2019, Proceedings. LNCS, vol. 11388, pp. 136--159. Springer (2019). \doi{10.1007/978-3-030-11245-5\_7}, \url{https://doi.org/10.1007/978-3-030-11245-5\_7}

\bibitem{DBLP:conf/cav/BradleyMS05}
Bradley, A.R., Manna, Z., Sipma, H.B.: Linear ranking with reachability. In: {CAV}. LNCS, vol.~3576, pp. 491--504. Springer (2005). \doi{10.1007/11513988\_48}, \url{https://doi.org/10.1007/11513988\_48}

\bibitem{IKOS}
Brat, G., Navas, J.A., Shi, N., Venet, A.: Ikos: A framework for static analysis based on abstract interpretation. In: Software Engineering and Formal Methods: 12th International Conference, SEFM 2014, Grenoble, France, September 1-5, 2014. Proceedings 12. pp. 271--277. Springer (2014)

\bibitem{DBLP:journals/jacm/CalcagnoDOY11}
Calcagno, C., Distefano, D., O'Hearn, P.W., Yang, H.: Compositional shape analysis by means of bi-abduction. J. {ACM}  \textbf{58}(6),  26:1--26:66 (2011). \doi{10.1145/2049697.2049700}, \url{https://doi.org/10.1145/2049697.2049700}

\bibitem{DBLP:conf/cav/ChakarovS13}
Chakarov, A., Sankaranarayanan, S.: Probabilistic program analysis with martingales. In: {CAV}. LNCS, vol.~8044, pp. 511--526. Springer (2013). \doi{10.1007/978-3-642-39799-8\_34}, \url{https://doi.org/10.1007/978-3-642-39799-8\_34}

\bibitem{DBLP:journals/toplas/ChatterjeeFG19}
Chatterjee, K., Fu, H., Goharshady, A.K.: Non-polynomial worst-case analysis of recursive programs. {ACM} Trans. Program. Lang. Syst.  \textbf{41}(4),  20:1--20:52 (2019). \doi{10.1145/3339984}, \url{https://doi.org/10.1145/3339984}

\bibitem{DBLP:conf/pldi/Chatterjee0GG20}
Chatterjee, K., Fu, H., Goharshady, A.K., Goharshady, E.K.: Polynomial invariant generation for non-deterministic recursive programs. In: {PLDI}. pp. 672--687. {ACM} (2020). \doi{10.1145/3385412.3385969}, \url{https://doi.org/10.1145/3385412.3385969}

\bibitem{DBLP:conf/ictac/ChenXYZZ07}
Chen, Y., Xia, B., Yang, L., Zhan, N., Zhou, C.: Discovering non-linear ranking functions by solving semi-algebraic systems. In: {ICTAC}. LNCS, vol.~4711, pp. 34--49. Springer (2007). \doi{10.1007/978-3-540-75292-9\_3}, \url{https://doi.org/10.1007/978-3-540-75292-9\_3}

\bibitem{DBLP:conf/cav/ChenHWZ15}
Chen, Y., Hong, C., Wang, B., Zhang, L.: Counterexample-guided polynomial loop invariant generation by lagrange interpolation. In: {CAV}. LNCS, vol.~9206, pp. 658--674. Springer (2015). \doi{10.1007/978-3-319-21690-4\_44}, \url{https://doi.org/10.1007/978-3-319-21690-4\_44}

\bibitem{ClangStaticAnalyzer}
Clang static analyzer: A source code analysis tool that finds bugs in c, c++, and objective-c programs. \url{https://clang-analyzer.llvm.org/} (2022)

\bibitem{DBLP:conf/cav/ColonSS03}
Col{\'{o}}n, M., Sankaranarayanan, S., Sipma, H.: Linear invariant generation using non-linear constraint solving. In: {CAV}. LNCS, vol.~2725, pp. 420--432. Springer (2003). \doi{10.1007/978-3-540-45069-6\_39}, \url{https://doi.org/10.1007/978-3-540-45069-6\_39}

\bibitem{DBLP:conf/tacas/ColonS01}
Col{\'{o}}n, M., Sipma, H.: Synthesis of linear ranking functions. In: {TACAS}. LNCS, vol.~2031, pp. 67--81. Springer (2001). \doi{10.1007/3-540-45319-9\_6}, \url{https://doi.org/10.1007/3-540-45319-9\_6}

\bibitem{DBLP:conf/vmcai/Cousot05}
Cousot, P.: Proving program invariance and termination by parametric abstraction, lagrangian relaxation and semidefinite programming. In: {VMCAI}. LNCS, vol.~3385, pp. 1--24. Springer (2005). \doi{10.1007/978-3-540-30579-8\_1}, \url{https://doi.org/10.1007/978-3-540-30579-8\_1}

\bibitem{DBLP:conf/popl/CousotC77}
Cousot, P., Cousot, R.: Abstract interpretation: {A} unified lattice model for static analysis of programs by construction or approximation of fixpoints. In: {POPL}. pp. 238--252. {ACM} (1977). \doi{10.1145/512950.512973}, \url{https://doi.org/10.1145/512950.512973}

\bibitem{ASTree}
Cousot, P., Cousot, R., Feret, J., Mauborgne, L., Min{\'e}, A., Monniaux, D., Rival, X.: The astr{\'e}e analyzer. In: Programming Languages and Systems: 14th European Symposium on Programming, ESOP 2005, Held as Part of the Joint European Conferences on Theory and Practice of Software, ETAPS 2005, Edinburgh, UK, April 4-8, 2005. Proceedings 14. pp. 21--30. Springer (2005)

\bibitem{DBLP:conf/popl/CousotH78}
Cousot, P., Halbwachs, N.: Automatic discovery of linear restraints among variables of a program. In: {POPL}. pp. 84--96. {ACM} Press (1978). \doi{10.1145/512760.512770}, \url{https://doi.org/10.1145/512760.512770}

\bibitem{CPAchecker}
Cpachecker: The configurable software-verification platform. \url{https://cpachecker.sosy-lab.org} (2022)

\bibitem{DBLP:conf/icse/CsallnerTS08}
Csallner, C., Tillmann, N., Smaragdakis, Y.: Dysy: dynamic symbolic execution for invariant inference. In: {ICSE}. pp. 281--290. {ACM} (2008). \doi{10.1145/1368088.1368127}, \url{https://doi.org/10.1145/1368088.1368127}

\bibitem{DBLP:journals/pacmpl/CyphertBKR19}
Cyphert, J., Breck, J., Kincaid, Z., Reps, T.W.: Refinement of path expressions for static analysis. Proc. {ACM} Program. Lang.  \textbf{3}({POPL}),  45:1--45:29 (2019). \doi{10.1145/3290358}, \url{https://doi.org/10.1145/3290358}

\bibitem{SVCOMP2023Veriabs}
Darke, P., Agrawal, S., Venkatesh, R.: Veriabs: A tool for scalable verification by abstraction (competition contribution). In: Tools and Algorithms for the Construction and Analysis of Systems: 27th International Conference, TACAS 2021, Held as Part of the European Joint Conferences on Theory and Practice of Software, ETAPS 2021, Luxembourg City, Luxembourg, March 27--April 1, 2021, Proceedings, Part II 27. pp. 458--462. Springer (2021)

\bibitem{DBLP:conf/fm/DavidKKL16}
David, C., Kesseli, P., Kroening, D., Lewis, M.: Danger invariants. In: {FM}. LNCS, vol.~9995, pp. 182--198 (2016). \doi{10.1007/978-3-319-48989-6\_12}, \url{https://doi.org/10.1007/978-3-319-48989-6\_12}

\bibitem{DBLP:conf/oopsla/DilligDLM13}
Dillig, I., Dillig, T., Li, B., McMillan, K.L.: Inductive invariant generation via abductive inference. In: {OOPSLA}. pp. 443--456. {ACM} (2013). \doi{10.1145/2509136.2509511}, \url{https://doi.org/10.1145/2509136.2509511}

\bibitem{DBLP:conf/sas/DonaldsonHKR11}
Donaldson, A.F., Haller, L., Kroening, D., R{\"{u}}mmer, P.: Software verification using k-induction. In: Yahav, E. (ed.) {SAS}. LNCS, vol.~6887, pp. 351--368. Springer (2011). \doi{10.1007/978-3-642-23702-7\_26}, \url{https://doi.org/10.1007/978-3-642-23702-7\_26}

\bibitem{FarkasLemma}
Farkas, J.: A fourier-f\'{e}le mechanikai elv alkalmaz\'{a}sai ({H}ungarian). Mathematikai\'{e}s Term\'{e}szettudom\'{a}nyi \'{E}rtesit\"{o}  \textbf{12},  457--472 (1894)

\bibitem{DBLP:conf/fmcad/FarzanK15}
Farzan, A., Kincaid, Z.: Compositional recurrence analysis. In: {FMCAD}. pp. 57--64. {IEEE} (2015)

\bibitem{DBLP:conf/cav/GanX0ZD20}
Gan, T., Xia, B., Xue, B., Zhan, N., Dai, L.: Nonlinear craig interpolant generation. In: {CAV}. LNCS, vol. 12224, pp. 415--438. Springer (2020). \doi{10.1007/978-3-030-53288-8\_20}, \url{https://doi.org/10.1007/978-3-030-53288-8\_20}

\bibitem{DBLP:conf/cav/0001LMN14}
Garg, P., L{\"{o}}ding, C., Madhusudan, P., Neider, D.: {ICE:} {A} robust framework for learning invariants. In: {CAV}. LNCS, vol.~8559, pp. 69--87. Springer (2014). \doi{10.1007/978-3-319-08867-9\_5}, \url{https://doi.org/10.1007/978-3-319-08867-9\_5}

\bibitem{DBLP:conf/popl/0001NMR16}
Garg, P., Neider, D., Madhusudan, P., Roth, D.: Learning invariants using decision trees and implication counterexamples. In: {POPL}. pp. 499--512. {ACM} (2016). \doi{10.1145/2837614.2837664}, \url{https://doi.org/10.1145/2837614.2837664}

\bibitem{DBLP:conf/sas/GopanR07}
Gopan, D., Reps, T.W.: Guided static analysis. In: Nielson, H.R., Fil{\'{e}}, G. (eds.) Static Analysis, 14th International Symposium, {SAS} 2007, Kongens Lyngby, Denmark, August 22-24, 2007, Proceedings. LNCS, vol.~4634, pp. 349--365. Springer (2007). \doi{10.1007/978-3-540-74061-2\_22}, \url{https://doi.org/10.1007/978-3-540-74061-2\_22}

\bibitem{DBLP:conf/pldi/GulwaniJK09}
Gulwani, S., Jain, S., Koskinen, E.: Control-flow refinement and progress invariants for bound analysis. In: Hind, M., Diwan, A. (eds.) Proceedings of the 2009 {ACM} {SIGPLAN} Conference on Programming Language Design and Implementation, {PLDI} 2009, Dublin, Ireland, June 15-21, 2009. pp. 375--385. {ACM} (2009). \doi{10.1145/1542476.1542518}, \url{https://doi.org/10.1145/1542476.1542518}

\bibitem{DBLP:conf/pldi/GulwaniSV08}
Gulwani, S., Srivastava, S., Venkatesan, R.: Program analysis as constraint solving. In: {PLDI}. pp. 281--292. {ACM} (2008). \doi{10.1145/1375581.1375616}, \url{https://doi.org/10.1145/1375581.1375616}

\bibitem{DBLP:conf/cav/GuptaR09}
Gupta, A., Rybalchenko, A.: Invgen: An efficient invariant generator. In: {CAV}. LNCS, vol.~5643, pp. 634--640. Springer (2009). \doi{10.1007/978-3-642-02658-4\_48}, \url{https://doi.org/10.1007/978-3-642-02658-4\_48}

\bibitem{DBLP:conf/pldi/HeSPV20}
He, J., Singh, G., P{\"{u}}schel, M., Vechev, M.T.: Learning fast and precise numerical analysis. In: {PLDI}. pp. 1112--1127. {ACM} (2020). \doi{10.1145/3385412.3386016}, \url{https://doi.org/10.1145/3385412.3386016}

\bibitem{DBLP:journals/entcs/HenryMM12}
Henry, J., Monniaux, D., Moy, M.: {PAGAI:} {A} path sensitive static analyser. Electron. Notes Theor. Comput. Sci.  \textbf{289},  15--25 (2012). \doi{10.1016/j.entcs.2012.11.003}, \url{https://doi.org/10.1016/j.entcs.2012.11.003}

\bibitem{DBLP:conf/lics/HrushovskiOP018}
Hrushovski, E., Ouaknine, J., Pouly, A., Worrell, J.: Polynomial invariants for affine programs. In: {LICS}. pp. 530--539. {ACM} (2018). \doi{10.1145/3209108.3209142}, \url{https://doi.org/10.1145/3209108.3209142}

\bibitem{DBLP:conf/issac/HumenbergerJK17}
Humenberger, A., Jaroschek, M., Kov{\'{a}}cs, L.: Automated generation of non-linear loop invariants utilizing hypergeometric sequences. In: {ISSAC}. pp. 221--228. {ACM} (2017). \doi{10.1145/3087604.3087623}, \url{https://doi.org/10.1145/3087604.3087623}

\bibitem{DBLP:conf/cav/JiFFC22}
Ji, Y., Fu, H., Fang, B., Chen, H.: Affine loop invariant generation via matrix algebra. In: Shoham, S., Vizel, Y. (eds.) Computer Aided Verification - 34th International Conference, {CAV} 2022, Haifa, Israel, August 7-10, 2022, Proceedings, Part {I}. Lecture Notes in Computer Science, vol. 13371, pp. 257--281. Springer (2022). \doi{10.1007/978-3-031-13185-1\_13}, \url{https://doi.org/10.1007/978-3-031-13185-1\_13}

\bibitem{DBLP:journals/pacmpl/KSG22}
K., H.G.V., Shoham, S., Gurfinkel, A.: Solving constrained horn clauses modulo algebraic data types and recursive functions. Proc. {ACM} Program. Lang.  \textbf{6}({POPL}),  1--29 (2022). \doi{10.1145/3498722}, \url{https://doi.org/10.1145/3498722}

\bibitem{DBLP:conf/dagstuhl/Kapur05}
Kapur, D.: Automatically generating loop invariants using quantifier elimination. In: Deduction and Applications. Dagstuhl Seminar Proceedings, vol. 05431. Internationales Begegnungs- und Forschungszentrum f{\"{u}}r Informatik (IBFI), Schloss Dagstuhl, Germany (2005), \url{http://drops.dagstuhl.de/opus/volltexte/2006/511}

\bibitem{DBLP:conf/pldi/KincaidBBR17}
Kincaid, Z., Breck, J., Boroujeni, A.F., Reps, T.W.: Compositional recurrence analysis revisited. In: {PLDI}. pp. 248--262. {ACM} (2017). \doi{10.1145/3062341.3062373}, \url{https://doi.org/10.1145/3062341.3062373}

\bibitem{DBLP:journals/pacmpl/KincaidCBR18}
Kincaid, Z., Cyphert, J., Breck, J., Reps, T.W.: Non-linear reasoning for invariant synthesis. Proc. {ACM} Program. Lang.  \textbf{2}({POPL}),  54:1--54:33 (2018). \doi{10.1145/3158142}, \url{https://doi.org/10.1145/3158142}

\bibitem{DBLP:conf/vmcai/LarrazRR13}
Larraz, D., Rodr{\'{\i}}guez{-}Carbonell, E., Rubio, A.: Smt-based array invariant generation. In: Giacobazzi, R., Berdine, J., Mastroeni, I. (eds.) Verification, Model Checking, and Abstract Interpretation, 14th International Conference, {VMCAI} 2013, Rome, Italy, January 20-22, 2013. Proceedings. Lecture Notes in Computer Science, vol.~7737, pp. 169--188. Springer (2013). \doi{10.1007/978-3-642-35873-9\_12}, \url{https://doi.org/10.1007/978-3-642-35873-9\_12}

\bibitem{DBLP:conf/pldi/LeZN19}
Le, T.C., Zheng, G., Nguyen, T.: {SLING:} using dynamic analysis to infer program invariants in separation logic. In: McKinley, K.S., Fisher, K. (eds.) Proceedings of the 40th {ACM} {SIGPLAN} Conference on Programming Language Design and Implementation, {PLDI} 2019, Phoenix, AZ, USA, June 22-26, 2019. pp. 788--801. {ACM} (2019). \doi{10.1145/3314221.3314634}, \url{https://doi.org/10.1145/3314221.3314634}

\bibitem{DBLP:journals/fcsc/LinWYZ14}
Lin, W., Wu, M., Yang, Z., Zeng, Z.: Proving total correctness and generating preconditions for loop programs via symbolic-numeric computation methods. Frontiers Comput. Sci.  \textbf{8}(2),  192--202 (2014)

\bibitem{DBLP:conf/tase/LinZCSXLS21}
Lin, Y., Zhang, Y., Chen, S., Song, F., Xie, X., Li, X., Sun, L.: Inferring loop invariants for multi-path loops. In: International Symposium on Theoretical Aspects of Software Engineering, {TASE} 2021, Shanghai, China, August 25-27, 2021. pp. 63--70. {IEEE} (2021). \doi{10.1109/TASE52547.2021.00030}, \url{https://doi.org/10.1109/TASE52547.2021.00030}

\bibitem{oopsla22/scalable}
Liu, H., Fu, H., Yu, Z., Song, J., Li, G.: Scalable linear invariant generation with {F}arkas' lemma. Proc. ACM Program. Lang.  \textbf{6}(OOPSLA2) (oct 2022). \doi{10.1145/3563295}, \url{https://doi.org/10.1145/3563295}

\bibitem{DBLP:books/daglib/0080029}
Manna, Z., Pnueli, A.: Temporal verification of reactive systems - safety. Springer (1995)

\bibitem{DBLP:conf/tacas/McMillan08}
McMillan, K.L.: Quantified invariant generation using an interpolating saturation prover. In: Ramakrishnan, C.R., Rehof, J. (eds.) {TACAS}. LNCS, vol.~4963, pp. 413--427. Springer (2008). \doi{10.1007/978-3-540-78800-3\_31}, \url{https://doi.org/10.1007/978-3-540-78800-3\_31}

\bibitem{DBLP:conf/esop/Mine04}
Min{\'{e}}, A.: Relational abstract domains for the detection of floating-point run-time errors. In: Schmidt, D.A. (ed.) Programming Languages and Systems, 13th European Symposium on Programming, {ESOP} 2004, Held as Part of the Joint European Conferences on Theory and Practice of Software, {ETAPS} 2004, Barcelona, Spain, March 29 - April 2, 2004, Proceedings. Lecture Notes in Computer Science, vol.~2986, pp. 3--17. Springer (2004). \doi{10.1007/978-3-540-24725-8\_2}, \url{https://doi.org/10.1007/978-3-540-24725-8\_2}

\bibitem{DBLP:conf/vmcai/Mine06}
Min{\'{e}}, A.: Symbolic methods to enhance the precision of numerical abstract domains. In: Emerson, E.A., Namjoshi, K.S. (eds.) Verification, Model Checking, and Abstract Interpretation, 7th International Conference, {VMCAI} 2006, Charleston, SC, USA, January 8-10, 2006, Proceedings. Lecture Notes in Computer Science, vol.~3855, pp. 348--363. Springer (2006). \doi{10.1007/11609773\_23}, \url{https://doi.org/10.1007/11609773\_23}

\bibitem{DBLP:conf/icse/NguyenKWF12}
Nguyen, T., Kapur, D., Weimer, W., Forrest, S.: Using dynamic analysis to discover polynomial and array invariants. In: {ICSE}. pp. 683--693. {IEEE} Computer Society (2012). \doi{10.1109/ICSE.2012.6227149}, \url{https://doi.org/10.1109/ICSE.2012.6227149}

\bibitem{ICSE2022dig}
Nguyen, T., Nguyen, K., Duong, H.: Syminfer: inferring numerical invariants using symbolic states. In: Proceedings of the ACM/IEEE 44th International Conference on Software Engineering: Companion Proceedings. pp. 197--201 (2022)

\bibitem{DBLP:conf/atva/OliveiraBP16}
de~Oliveira, S., Bensalem, S., Prevosto, V.: Polynomial invariants by linear algebra. In: {ATVA}. LNCS, vol.~9938, pp. 479--494 (2016). \doi{10.1007/978-3-319-46520-3\_30}, \url{https://doi.org/10.1007/978-3-319-46520-3\_30}

\bibitem{DBLP:conf/atva/OliveiraBP17}
de~Oliveira, S., Bensalem, S., Prevosto, V.: Synthesizing invariants by solving solvable loops. In: {ATVA}. LNCS, vol. 10482, pp. 327--343. Springer (2017). \doi{10.1007/978-3-319-68167-2\_22}, \url{https://doi.org/10.1007/978-3-319-68167-2\_22}

\bibitem{DBLP:conf/pldi/PadonMPSS16}
Padon, O., McMillan, K.L., Panda, A., Sagiv, M., Shoham, S.: Ivy: safety verification by interactive generalization. In: {PLDI}. pp. 614--630. {ACM} (2016). \doi{10.1145/2908080.2908118}, \url{https://doi.org/10.1145/2908080.2908118}

\bibitem{DBLP:conf/vmcai/PodelskiR04}
Podelski, A., Rybalchenko, A.: A complete method for the synthesis of linear ranking functions. In: {VMCAI}. LNCS, vol.~2937, pp. 239--251. Springer (2004). \doi{10.1007/978-3-540-24622-0\_20}, \url{https://doi.org/10.1007/978-3-540-24622-0\_20}

\bibitem{FSE2022}
Riley, D., Fedyukovich, G.: Multi-phase invariant synthesis. In: Proceedings of the 30th ACM Joint European Software Engineering Conference and Symposium on the Foundations of Software Engineering. pp. 607--619 (2022)

\bibitem{DBLP:conf/sas/Rodriguez-CarbonellK04}
Rodr{\'{\i}}guez{-}Carbonell, E., Kapur, D.: An abstract interpretation approach for automatic generation of polynomial invariants. In: {SAS}. LNCS, vol.~3148, pp. 280--295. Springer (2004). \doi{10.1007/978-3-540-27864-1\_21}, \url{https://doi.org/10.1007/978-3-540-27864-1\_21}

\bibitem{DBLP:conf/issac/Rodriguez-CarbonellK04}
Rodr{\'{\i}}guez{-}Carbonell, E., Kapur, D.: {A}utomatic {G}eneration of {P}olynomial {L}oop {I}nvariants: {A}lgebraic {F}oundations. In: {ISSAC}. pp. 266--273. {ACM} (2004). \doi{10.1145/1005285.1005324}, \url{https://doi.org/10.1145/1005285.1005324}

\bibitem{DBLP:conf/iclr/RyanWYGJ20}
Ryan, G., Wong, J., Yao, J., Gu, R., Jana, S.: {CLN2INV:} learning loop invariants with continuous logic networks. In: 8th International Conference on Learning Representations, {ICLR} 2020, Addis Ababa, Ethiopia, April 26-30, 2020. OpenReview.net (2020), \url{https://openreview.net/forum?id=HJlfuTEtvB}

\bibitem{DBLP:conf/popl/SankaranarayananSM04}
Sankaranarayanan, S., Sipma, H., Manna, Z.: Non-linear loop invariant generation using gr{\"{o}}bner bases. In: {POPL}. pp. 318--329. {ACM} (2004). \doi{10.1145/964001.964028}, \url{https://doi.org/10.1145/964001.964028}

\bibitem{DBLP:conf/sas/SankaranarayananSM04}
Sankaranarayanan, S., Sipma, H.B., Manna, Z.: Constraint-based linear-relations analysis. In: {SAS}. LNCS, vol.~3148, pp. 53--68. Springer (2004). \doi{10.1007/978-3-540-27864-1\_7}, \url{https://doi.org/10.1007/978-3-540-27864-1\_7}

\bibitem{DBLP:books/daglib/0090562}
Schrijver, A.: Theory of linear and integer programming. Wiley-Interscience series in discrete mathematics and optimization, Wiley (1999)

\bibitem{DBLP:journals/fmsd/SharmaA16}
Sharma, R., Aiken, A.: From invariant checking to invariant inference using randomized search. Formal Methods Syst. Des.  \textbf{48}(3),  235--256 (2016). \doi{10.1007/s10703-016-0248-5}, \url{https://doi.org/10.1007/s10703-016-0248-5}

\bibitem{DBLP:conf/cav/SharmaDDA11}
Sharma, R., Dillig, I., Dillig, T., Aiken, A.: Simplifying loop invariant generation using splitter predicates. In: Gopalakrishnan, G., Qadeer, S. (eds.) Computer Aided Verification - 23rd International Conference, {CAV} 2011, Snowbird, UT, USA, July 14-20, 2011. Proceedings. Lecture Notes in Computer Science, vol.~6806, pp. 703--719. Springer (2011). \doi{10.1007/978-3-642-22110-1\_57}, \url{https://doi.org/10.1007/978-3-642-22110-1\_57}

\bibitem{DBLP:conf/esop/0001GHALN13}
Sharma, R., Gupta, S., Hariharan, B., Aiken, A., Liang, P., Nori, A.V.: A data driven approach for algebraic loop invariants. In: {ESOP}. LNCS, vol.~7792, pp. 574--592. Springer (2013). \doi{10.1007/978-3-642-37036-6\_31}, \url{https://doi.org/10.1007/978-3-642-37036-6\_31}

\bibitem{DIG_CIVL}
Siegel, S.F., Zheng, M., Luo, Z., Zirkel, T.K., Marianiello, A.V., Edenhofner, J.G., Dwyer, M.B., Rogers, M.S.: Civl: the concurrency intermediate verification language. In: Proceedings of the International Conference for High Performance Computing, Networking, Storage and Analysis. pp. 1--12 (2015)

\bibitem{DBLP:conf/cav/SilvermanK19}
Silverman, J., Kincaid, Z.: Loop summarization with rational vector addition systems. In: Dillig, I., Tasiran, S. (eds.) Computer Aided Verification - 31st International Conference, {CAV} 2019, New York City, NY, USA, July 15-18, 2019, Proceedings, Part {II}. Lecture Notes in Computer Science, vol. 11562, pp. 97--115. Springer (2019). \doi{10.1007/978-3-030-25543-5\_7}, \url{https://doi.org/10.1007/978-3-030-25543-5\_7}

\bibitem{DBLP:conf/popl/SinghPV17}
Singh, G., P{\"{u}}schel, M., Vechev, M.T.: Fast polyhedra abstract domain. In: Castagna, G., Gordon, A.D. (eds.) Proceedings of the 44th {ACM} {SIGPLAN} Symposium on Principles of Programming Languages, {POPL} 2017, Paris, France, January 18-20, 2017. pp. 46--59. {ACM} (2017)

\bibitem{DBLP:conf/fmcad/SomenziB11}
Somenzi, F., Bradley, A.R.: {IC3:} where monolithic and incremental meet. In: Bjesse, P., Slobodov{\'{a}}, A. (eds.) International Conference on Formal Methods in Computer-Aided Design, {FMCAD} '11, Austin, TX, USA, October 30 - November 02, 2011. pp.~3--8. {FMCAD} Inc. (2011), \url{http://dl.acm.org/citation.cfm?id=2157657}

\bibitem{DBLP:conf/pldi/SrivastavaG09}
Srivastava, S., Gulwani, S.: Program verification using templates over predicate abstraction. In: Hind, M., Diwan, A. (eds.) Proceedings of the 2009 {ACM} {SIGPLAN} Conference on Programming Language Design and Implementation, {PLDI} 2009, Dublin, Ireland, June 15-21, 2009. pp. 223--234. {ACM} (2009). \doi{10.1145/1542476.1542501}, \url{https://doi.org/10.1145/1542476.1542501}

\bibitem{Sting}
Sting: Stanford invariant generator. \url{http://theory.stanford.edu/~srirams/Software/sting.html} (2006)

\bibitem{svcomp}
Software verification competition. \url{https://sv-comp.sosy-lab.org} (2023)

\bibitem{tarjan1972depth}
Tarjan, R.: Depth-first search and linear graph algorithms. SIAM journal on computing  \textbf{1}(2),  146--160 (1972)

\bibitem{oopsla23}
Wang, C., Lin, F.: Solving conditional linear recurrences for program verification: The periodic case. In: {OOPSLA}. ACM (2023), to appear

\bibitem{DBLP:conf/pldi/WangS0CG21}
Wang, J., Sun, Y., Fu, H., Chatterjee, K., Goharshady, A.K.: Quantitative analysis of assertion violations in probabilistic programs. In: {PLDI}. pp. 1171--1186. {ACM} (2021). \doi{10.1145/3453483.3454102}, \url{https://doi.org/10.1145/3453483.3454102}

\bibitem{AutoSpec}
Wen, C., Cao, J., Su, J., Xu, Z., Qin, S., He, M., Li, H., Cheung, S.C., Tian, C.: Enchanting program specification synthesis by large language models using static analysis and program verification. In: International Conference on Computer Aided Verification. pp. 302--328. Springer (2024)

\bibitem{DBLP:conf/sigsoft/XieCLLL16}
Xie, X., Chen, B., Liu, Y., Le, W., Li, X.: Proteus: computing disjunctive loop summary via path dependency analysis. In: Zimmermann, T., Cleland{-}Huang, J., Su, Z. (eds.) Proceedings of the 24th {ACM} {SIGSOFT} International Symposium on Foundations of Software Engineering, {FSE} 2016, Seattle, WA, USA, November 13-18, 2016. pp. 61--72. {ACM} (2016). \doi{10.1145/2950290.2950340}, \url{https://doi.org/10.1145/2950290.2950340}

\bibitem{DBLP:conf/sigsoft/Xu0W20}
Xu, R., He, F., Wang, B.: Interval counterexamples for loop invariant learning. In: {ESEC/FSE}. pp. 111--122. {ACM} (2020). \doi{10.1145/3368089.3409752}, \url{https://doi.org/10.1145/3368089.3409752}

\bibitem{DBLP:journals/fcsc/YangZZX10}
Yang, L., Zhou, C., Zhan, N., Xia, B.: Recent advances in program verification through computer algebra. Frontiers Comput. Sci. China  \textbf{4}(1),  1--16 (2010). \doi{10.1007/s11704-009-0074-7}, \url{https://doi.org/10.1007/s11704-009-0074-7}

\bibitem{DBLP:conf/pldi/YaoRWJG20}
Yao, J., Ryan, G., Wong, J., Jana, S., Gu, R.: Learning nonlinear loop invariants with gated continuous logic networks. In: {PLDI}. pp. 106--120. {ACM} (2020). \doi{10.1145/3385412.3385986}, \url{https://doi.org/10.1145/3385412.3385986}

\bibitem{z3}
Z3. \url{https://github.com/Z3Prover/z3} (2023)

\end{thebibliography}

\clearpage
\appendix

\lstset{language=program}
\lstset{tabsize=3}
\newsavebox{\unnestedP}
\begin{lrbox}{\unnestedP}
\begin{lstlisting}[mathescape]
switch {
  case $\phi_{P,1}$: $\mathbf{x}:=F_{P,1}(\mathbf{x})$;$\delta_{P,1}$;
  $\cdots$
  case $\phi_{P,p}$: $\mathbf{x}:=F_{P,p}(\mathbf{x})$;$\delta_{P,p}$;
}
\end{lstlisting}
\end{lrbox}

\lstset{language=program}
\lstset{tabsize=3}
\newsavebox{\unnestedQ}
\begin{lrbox}{\unnestedQ}
\begin{lstlisting}[mathescape]
switch {
  case $\phi_{Q,1}$: $\mathbf{x}:=\mathbf{F}_{Q,1}(\mathbf{x})$;$\delta_{Q,1}$;
  $\cdots$
  case $\phi_{Q,q}$: $\mathbf{x}:=\mathbf{F}_{Q,q}(\mathbf{x})$;$\delta_{Q,q}$;
}
\end{lstlisting}
\end{lrbox}
% \begin{figure}[t]
% \centering
% \usebox{\unnestedQ}
% \caption{The transformation for a code segment $Q$}\label{fig:unnestedQ}
% \end{figure}

\lstset{language=program}
\lstset{tabsize=3}
\newsavebox{\unnestedRsequential}
\begin{lrbox}{\unnestedRsequential}
\begin{lstlisting}[mathescape]
switch {
  $\cdots$
  case $\phi_{P,i}$: 
    $\mathbf{x}:=(\mathbf{F}_{P,i}(\mathbf{x}))$;
    $break$;($\mbox{if }\delta_{P,i}=\textbf{break}$)
  $\cdots$
  case $\phi_{P,i}\wedge \phi_{Q,j}[\mathbf{F}_{P,i}(\mathbf{x})/\mathbf{x}]$: 
    $\mathbf{x}:=\mathbf{F}_{Q,j} (\mathbf{F}_{P,i}(\mathbf{x}))$; 
    $\delta_{Q,j}$;($\mbox{if }\delta_{P,i}=\textbf{skip}$)
  $\cdots$
}
\end{lstlisting}
\end{lrbox}
% \begin{figure}[t]
% \centering
% \usebox{\unnestedRsequential}
% \caption{The transformed code a sequential composition}\label{fig:unnestedRsequential}
% \end{figure}

\lstset{language=program}
\lstset{tabsize=3}
\newsavebox{\unnestedRconditional}
\begin{lrbox}{\unnestedRconditional}
\begin{lstlisting}[mathescape]
switch {
  $\cdots$
  case $\phi_{P,i}\wedge b$: 
    $\mathbf{x}:=\mathbf{F}_{P,i}(\mathbf{x})$;$\delta_{P,i}$;
  $\cdots$
  case $\phi_{Q,j}\wedge \neg b$: 
    $\mathbf{x}:=\mathbf{F}_{Q,j}(\mathbf{x})$;$\delta_{Q,j}$;
  $\cdots$
}
\end{lstlisting}
\end{lrbox}
\section{Process of transformation to canonical form}
\label{appendix:transform}
Here, we provide a detailed demonstration of how to transform the loop body $P$ of a non-nested affine program into its canonical form:

\begin{figure}[h]
\begin{minipage}[b]{\linewidth} 
\centering
  \begin{minipage}[b]{0.3\linewidth}
    \begin{minipage}[b]{\linewidth}
      \centering
      \resizebox{\linewidth}{!}{
      \usebox{\unnestedP}
      }
    \subcaption{$\mathsf{C}_P$}
    \label{fig:unnestedP}
    \end{minipage}
    \begin{minipage}[b]{\linewidth}
        \centering
        \resizebox{\linewidth}{!}{
        \usebox{\unnestedQ}
        }
    \subcaption{$\mathsf{C}_Q$}
    \label{fig:unnestedQ}
    \end{minipage}
  \end{minipage}
  \hfill 
  \begin{minipage}[b]{0.3\linewidth}
      \centering
      \resizebox{\linewidth}{!}{
      \usebox{\unnestedRsequential}
      }
  \subcaption{The sequential case}
  \label{fig:unnestedRsequential}
  \end{minipage}
  \hfill 
  \begin{minipage}[b]{0.3\linewidth}
      \centering
      \resizebox{\linewidth}{!}{
      \usebox{\unnestedRconditional}
      }
  \subcaption{The conditional case}
  \label{fig:unnestedRconditional}
  \end{minipage}
\caption{The canonical form of transformation (TF) for $P$, $Q$}
\label{fig:Morecase}
\end{minipage}
\end{figure}
\begin{itemize}
\item For the base case where the program $P$ is either a single affine assignment $\mathbf{x}:=\mathbf{F}(\mathbf{x})$ or resp. the \textbf{break} statement, the transformed program $\mathsf{C}_P$ is simply $\textbf{switch}~\{\textbf{case}~\mathbf{true}: \mathbf{x}:=\mathbf{F}(\mathbf{x}); 
\textbf{skip};\}$ or resp. $\textbf{switch}~\{\textbf{case}~\mathbf{true}: \mathbf{x}:=\mathbf{x}; \textbf{break};\}$, respectively.
\item For a sequential composition $R=P;Q$, the algorithm recursively computes $\mathsf{C}_P$ and $\mathsf{C}_Q$ as in Figure~\ref{fig:unnestedP} and Figure~\ref{fig:unnestedQ} respectively, and then compute $\mathsf{C}_R$ as in Figure~\ref{fig:unnestedRsequential} for which: 
\begin{itemize}
\item For each $1\le i\le p$ such that $\delta_{P,i}=\textbf{break}$, we have the branch $\mathbf{x}:=\mathbf{F}_{P,i}(\mathbf{x});\textbf{break};$ (i.e., the branch already breaks in the execution of $P$). 
\item For each $1\le i\le p$ and $1\le j\le q$ such that $\delta_{P,i}=\textbf{skip}$, we have the branch $\mathbf{x}:=\mathbf{F}_{Q,j} (\mathbf{F}_{P,i}(\mathbf{x})); \delta_{Q,j}$ under the branch condition $\phi_{P,i}\wedge (\phi_{Q,j}[\mathbf{F}_{P,i}(\mathbf{x})/\mathbf{x}])$ (i.e., the branch continues to the execution of $Q$). 
\end{itemize}
\item For a conditional branch $R=\textbf{if }b\textbf{ then }P\textbf{ else }Q$, the algorithm recursively computes $\mathsf{C}_P$ and $\mathsf{C}_Q$ as in the previous case, and then compute $\mathsf{C}_R$ as in Figure~\ref{fig:unnestedRconditional} for which: 
\begin{itemize}
\item For each $1\le i\le p$, we have the branch $\mathbf{x}=\mathbf{F}_{P,i}(\mathbf{x});\delta_{P,i};$ with branch condition $\phi_{P,i}\wedge b$ (i.e., the branch conditions of $P$ is conjuncted with the extra condition $b$). 
\item  For each $1\le j\le q$, we have the branch $\mathbf{x}=\mathbf{F}_{Q,j}(\mathbf{x});\delta_{Q,j};$ with branch condition $\phi_{Q,j}\wedge \neg b$ (i.e., the branch conditions of $Q$ is conjuncted with the extra condition $\neg b$).  
\end{itemize}
\end{itemize}

\section{Proof of Theorem~\ref{thr:infeasible}}
\label{app:thm_infeasible}
\par 
\textbf{Theorem}~\ref{thr:infeasible}. \textit{Let \(\Gamma\) be an ATS. For any AAM \(\eta\) that fulfills the initial and consecution conditions derived from the ATS \(\Gamma\) with the original constraints for the infeasible implication as in each consecution tabular of Figure~\ref{tab:farkascons} (aimed at \(-1>=0\)) with each $\mu$ in an infeasible implication instantiated as $k$ for some $k>0$, it is equivalent to set all $\mu$'s to $1$ while preserving the constraints of infeasible implication.}

To prove the theorem, for convenience, we denote the consecution tabular with $-1\geq 0$ as constraint consecution tabular and the consecution tabular without $-1\geq 0$ as transition consecution tabular. Then we prove that constraint consecution tabular is equivalent whether $\mu=1$ or $\mu=k,\forall k>0$.
\begin{figure}[ht]
    \centering
    \begin{minipage}[b]{0.6\linewidth}
    \centering
        \setlength{\arraycolsep}{0mm}{\small
        $$
        \begin{array}{c|rcccrcrcccrcrcc}
            k        &c_{\tsLoc, 1} x_1     &+  &\cdots &+ &c_{\tsLoc, n} x_n    &    &     &     &     &     &      &+  &d_{\tsLoc}  &\geq  &0  \\
        
            \lambda_{0}  &       &    &       &    &       &    &     &     &     &     &      &     &1  &\geq     &0      \\
            
            \lambda_{1}  &a_{11} x_1       &+ &\cdots &+ &a_{1n} x_n       &+                                     
                             &a_{11}' x_{1}'   &+ &\cdots &+ &a_{1n}' x_{n}'   &+  &b_1&\geq     &0      \\ 
            
            \vdots         &\vdots                &    &         &    &\vdots                &     
                             &\vdots                &    &         &    &\vdots                &     &\vdots         &   &      \\ 
            
            \lambda_{m}  &a_{m1} x_1       &+ &\cdots &+ &a_{mn} x_n       &+  
                             &a_{m1}' x_{1}'   &+ &\cdots &+ &a_{mn}' x_{n}'   &+  &b_m&\geq     &0      \\ 
            \hline
                             
                             &       &    &       &    &       &    &     &     &     &     &      &     &-1 &\geq  &0
        \end{array}
        $$
        }        
        \subcaption{$\mu=k,\forall k> 0$}
    \end{minipage}
    \caption{Constraint consecution tabular}
\end{figure}

We scale the leftmost coefficient column $k$ to be $1,\lambda'_i$'s by multiplying $\frac{1}{k}$, where $\lambda'_i=\frac{\lambda_i}{k}$. The coefficient of invariants after transformation is the same as the previous tabular.

\begin{figure}[h]
    \centering
    \begin{minipage}[b]{0.6\linewidth}
        \centering
        \setlength{\arraycolsep}{0mm}{\small
        $$
        \begin{array}{c|rcccrcrcccrcrcc}
            1        &c_{\tsLoc, 1} x_1     &+  &\cdots &+ &c_{\tsLoc, n} x_n    &    &     &     &     &     &      &+  &d_{\tsLoc}  &\geq  &0  \\
        
            \lambda'_{0}  &       &    &       &    &       &    &     &     &     &     &      &     &1  &\geq     &0      \\
            
            \lambda'_{1}  &a_{11} x_1       &+ &\cdots &+ &a_{1n} x_n       &+                                     
                             &a_{11}' x_{1}'   &+ &\cdots &+ &a_{1n}' x_{n}'   &+  &b_1&\geq     &0      \\ 
            
            \vdots         &\vdots                &    &         &    &\vdots                &     
                             &\vdots                &    &         &    &\vdots                &     &\vdots         &   &      \\ 
            
            \lambda'_{m}  &a_{m1} x_1       &+ &\cdots &+ &a_{mn} x_n       &+  
                             &a_{m1}' x_{1}'   &+ &\cdots &+ &a_{mn}' x_{n}'   &+  &b_m&\geq     &0      \\ 
            \hline
                             
                             &       &    &       &    &       &    &     &     &     &     &      &     &-\frac{1}{k} &\geq  &0
        \end{array}
        $$
        }        
        
    \end{minipage}
    \caption{Transformed constraint consecution tabular}

\end{figure}

Consider all the constraint consecution tabular and choose the maximum $k_{max}$ of their $k$'s. Then, we scale $\lambda_i$'s, $c_{l,i}$'s and $d_{\tsLoc}$ by $k_{max}$ and modify $\lambda'_0$ to be $\lambda''_0=\lambda'_0+\frac{k_{max}}{k}-1$. Note that it's necessary to select the fixed $k_{max}$ to scale $c_{l,i}$'s, so that we avoid affecting the solution in the transition consecution tabular as transition consecution tabular is always satisfied if we multiply $c$ with fixed constant $k_{max}$.

\begin{figure}[ht]
    \centering
    \begin{minipage}[b]{0.8\linewidth}
        \centering
        \setlength{\arraycolsep}{0mm}{\small
        $$\centering
        \begin{array}{c|rcccrcrcccrcrcc}
            1        &k_{max}\cdot c_{\tsLoc, 1} x_1     &+  &\cdots &+ &k_{max}\cdot c_{\tsLoc, n} x_n    &    &     &     &     &     &      &+  &k_{max}\cdot d_{\tsLoc}  &\geq  &0  \\
        
            \lambda''_{0}  &       &    &       &    &       &    &     &     &     &     &      &     &1  &\geq     &0      \\
            
            \lambda'_{1}  &a_{11} x_1       &+ &\cdots &+ &a_{1n} x_n       &+                                     
                             &a_{11}' x_{1}'   &+ &\cdots &+ &a_{1n}' x_{n}'   &+  &b_1&\geq     &0      \\ 
            
            \vdots         &\vdots                &    &         &    &\vdots                &     
                             &\vdots                &    &         &    &\vdots                &     &\vdots         &   &      \\ 
            
            \lambda'_{m}  &a_{m1} x_1       &+ &\cdots &+ &a_{mn} x_n       &+  
                             &a_{m1}' x_{1}'   &+ &\cdots &+ &a_{mn}' x_{n}'   &+  &b_m&\geq     &0      \\ 
            \hline
                             
                             &       &    &       &    &       &    &     &     &     &     &      &     &-1 &\geq  &0
        \end{array}
        $$
        }     
    \end{minipage}
    \caption{Equivalent constraint transformation tabular}
\end{figure}

Thus we prove there is no accuracy loss as we set $\mu=1$ by means of coefficient scaling.

In the implementation phase, we decompose the polyhedron derived from infeasible implications and subsequently union the resulting polytope and polyhedral cone into the final solution set. The validity of this methodology is further substantiated in Appendix~\ref{sec:appendix_minkowski}.

\section{Proof of correctness of solutions to invariant sets in the implementation part}\label{sec:appendix_minkowski}
\par
In the implementation, we utilize decomposition theorem of polyhedra and decompose the solution set of invariant when $\mu=1$ to be a polytope $P$ and a polyhedral cone $C$. Similarly, we denote $F$ as the solution set of invariants, which contains the coefficient of invariants at any locations and $F'$ as the solution set of invariants when $\mu=1$ in all the consecution tabular. Then the union of the polytope and polyhedral cone is chosen as our solution set of invariants $F^*=P\cup C$, where $F'=P+C$.

\begin{lemma}
    \textbf{Decomposition theorem of polyhedra.} A set $P$ of vectors in Euclidean space is a polyhedron if and only if $P=Q+C$ for some polytope $Q$ and some polyhedral cone $C$.
\end{lemma}

Now, we are going to prove the correctness of $F^*$. I.e., the vectors in polytope and polyhedral cone are both the coefficient of invariants in different locations.

Consider the relation between $F$ and $F'$, we define that $F''=\{k\cdot \boldsymbol{c}\mid \boldsymbol{c}\in F',k>0\}$.

\begin{proposition}
    $F=F''$
\end{proposition}
\begin{proof}
    We first consider $F\subseteq F''$, which equivalent to prove $\forall \boldsymbol{c}\in F, \exists \boldsymbol{c_0} \in F',k\in R$ such that $k\cdot \boldsymbol{c_0}=\boldsymbol{c}$. We consider the transition consecution tabular. Note that, the consecution tabular is always satisfied if we scale the invariant $\mathrm{\eta}(\tsLoc)$ and $\mathrm{\eta}(\tsLoc')$ simultaneously.

    Then consider the constraint consecution tabular. We have proved if we multiply $\boldsymbol{c_0}\in F'$ with $k_{max}$, we can find corresponding $\boldsymbol{c}=k_{max}*\boldsymbol{c_0}$ is the solution to constraint consecution tabular with $\mu=k,\forall k>0$. (the definition of $k_{max}$ and proof is given in Appendix~\ref{app:thm_infeasible})

    Thus, we prove that $\forall \boldsymbol{c}\in F$, there exists $\boldsymbol{c_0}\in F'$ and $k_{max}\in R$ such that $k_{max}*\boldsymbol{c_0}=\boldsymbol{c}$ and have $F\subseteq F''$. 

    Secondly, we prove $F''\subseteq F$. From the definition of $F''$, if $\boldsymbol{c}\in F'$, we multiply $\frac{1}{k}$ to $\mu=1$ in the constraint consecution tabular, and $k\cdot \boldsymbol{c}$ satisfy the transformed tabulars and other transition consecution tabulars. So $k\cdot \boldsymbol{c}\in F$, and we have $F''\subseteq F$.

    So $F=F''$.
\end{proof}

We utilize the decomposition theorem in $F'$ and have $F'=P+C$, where $P$ is a polytope and $C$ is a polyhedral cone. From the properties of polytope and polyhedral cone, a polytope is a convex hull of finitely many vectors and a polyhedral cone is finitely generated by some vectors.

\begin{equation}
    P=\{\boldsymbol{p_1},\boldsymbol{p_2},\ldots,\boldsymbol{p_n}\}
\end{equation}
\begin{equation}
    C=\{\boldsymbol{g_1},\boldsymbol{g_2},\ldots,\boldsymbol{g_m}\}
\end{equation}

Note that the addition in the theorem means Minkowski sum, Thus,
\begin{equation}
    F'=P+C=\{\boldsymbol{p_1},\boldsymbol{p_2},\ldots,\boldsymbol{p_n};\boldsymbol{g_1},\boldsymbol{g_2},\ldots,\boldsymbol{g_m}\}
\end{equation}

Where $\boldsymbol{p_i}$'s represents the vectors finitely generate the polytope $P$ and $g_i$'s represents the vectors finitely generate the polyhedral cone $C$. That means that $\forall v \in F'$, $v=a_1\boldsymbol{p_1}+\cdots+a_n \boldsymbol{p_n}+b_1 \boldsymbol{g_1}+\cdots + b_m \boldsymbol{g_m}$,where $\sum_i a_i=1$ (from the definition of convex hull) and $a_i,b_i\geq 0,\forall i$. Consider $F=F''=\{k\cdot \boldsymbol{c}\mid \exists k>0,k\cdot \boldsymbol{c}\in F,\boldsymbol{c}\in F'\}$, it's concluded that $\forall \boldsymbol{v} \in F, \boldsymbol{v}=a'_1\boldsymbol{p_1}+\cdots+a'_n \boldsymbol{p_n}+b'_1+\boldsymbol{g_1}+\cdots + b'_m \boldsymbol{g_m}$, where $a'_i=ka_i,k>0$ and $b'_i=kb_i,k>0$.

Thus, it's obvious that $\forall \boldsymbol{p}\in P$, we have $\boldsymbol{p}\in F$ as we can set $\boldsymbol{g_i}=0,\forall i$. However, we can not use the similar method to prove $\forall \boldsymbol{c} \in C,\boldsymbol{c}\in F$, as the $\sum_i a_i$ equal to a non-zero number and $a_i\geq 0, \forall i$. So to prove the $P\cup C$ is also the solution set of invariants, we should consider the practical implications of invariant.

We have known that $F=\{\boldsymbol{v}\mid \boldsymbol{v}=a'_1\boldsymbol{p_1}+\cdots+a'_n \boldsymbol{p_n}+b'_1 +_1+\cdots + b'_m \boldsymbol{g_m},a'_i\geq 0,\boldsymbol{g_i}\geq 0 \forall i,\sum_i a'_i>0\}$ corresponding to the solution of $\boldsymbol{v}^T x<=d_{\tsLoc}$.

Destruct $F$ to be 
\begin{equation}
    F=\{p+c\mid p=a'_1\boldsymbol{p_1}+\cdots+a'_n \boldsymbol{p_n},c=\boldsymbol{g_1}+\cdots + b'_m \boldsymbol{g_m},a'_i\geq 0,\boldsymbol{g_i}\geq 0 \forall i,\sum_i a'_i>0\}
\end{equation}

From the above conclusion, $P\subseteq F$, which means $\forall \boldsymbol{p} \in P$, $\boldsymbol{p}^T \boldsymbol{x}<=d_{\tsLoc}$ is satisfied. Also, $\forall \boldsymbol{v}\in F,\boldsymbol{v}=\boldsymbol{p}+\boldsymbol{c},\boldsymbol{p}\in P, \boldsymbol{c}\in C$, and $(\boldsymbol{p}+\boldsymbol{c})^T \boldsymbol{x}<=d_\tsLoc$. 

Consider $\forall \varepsilon>0$, we have $(\varepsilon \boldsymbol{p}+\boldsymbol{c})^T \boldsymbol{x}<=d_{\tsLoc}$. Thus, we finally conclude that $\lim_{\varepsilon \rightarrow 0} (\varepsilon \boldsymbol{p}+\boldsymbol{c})^T \boldsymbol{x}=\boldsymbol{c}^T \boldsymbol{x}<=d_{\tsLoc}$, which means for all $\boldsymbol{c}\in C$, $\boldsymbol{c}$ is also a solution to invariants. Thus we prove $C\in F$, and $P \cup C \in F$.

So it's correct to directly use the union of the polytope and the polyhedral cone to represents the solution set of invariants.
\clearpage
\section{Control Flow Transformation: Loop Summary}
\label{app:loop_sum}
In this section, we consider further optimizations. The first is the tackling of infeasible implications (i.e., "$-1\ge 0$") in the application of Farkas' Lemma. 
This situation has not been handled in the previous approaches in Farkas' Lemma. The infeasible implication is important in the generation of disjunctive invariants since ignoring them would break the internal disjunctive feature of the loop, thus leading to the failure of the generation of the desired disjunctive invariant. A key difficulty to tackle the infeasible implication is that we obtain general polyhedra rather than polyhedral cones in establishing the constraints for invariant generation, and thus cannot directly apply the generator computation over polyhedral cones. To address this difficulty, we show that it suffices to fix the nonlinear parameter $\mu$ multiplied to the template 
in the Farkas tabular (Figure~\ref{tab:farkas}) to $1$ and include the generators of both the polytope and the polyhedral cone of the Minkowski decomposition of the polyhedron resulting from the constraint solving of the PAP after the application of Farkas' Lemma. As its correctness proof is somewhat technical, we relegate them to Appendix~\ref{app:thm_infeasible} and Appendix~\ref{sec:appendix_minkowski}.

The second is the extension of our approach to nested loop. The main difficulty here is how to handle the inner loops in a nested loop. Recall that in the previous section, we transform a non-nested loop into a canonical form and further transforms it into an affine transition system. This cannot be applied to nested loops since the inner loops does not obey this canonical form. To address this difficulty, we consider to use the standard technique of loop summary to abstract the input-output relationship of the inner loops, and use the loop summaries of inner loops to construct the overall affine transition system for the outer loops. 

Given a nested affine while loop $W$, our approach works by first recursively computing the loop summary $\ProcSmry_{W'}$ for each inner while loop $W'$ in $W$ (from the innermost to the outermost), 
%that is a while loop which is in the loop body of $W$ and does not lie in other inner loops), 
and then %construct the affine transition system for 
tackling the main loop body via the control flow of the loop body and the loop summaries $\ProcSmry_{W'}$ of the inner loops. Below we fix a nested affine while loop $W$ with variable set $\tsVars=\{\tsVar_1,\dots,\tsVar_n\}$ and present the technical details. 

The most involved part in our approach is the transformation of the main loop $W$ into its corresponding \LTS{}.  %and the loop summaries of its inner loops. 
%Unlike the situation of unnested while loops, a direct recursive algorithm that transforms the loop $W$ into a canonical form in Figure~\ref{fig:unnestedPQandRecursive} as in the unnested case is not possible, since one needs to tackle the loop summaries from the inner while loops in $W$.
%since to apply the state-transition pattern here one needs to specify the conditions for the output of the inner while loops in a total conditionl branch (i.e., a conditional branch that appears at the top level like some $\varphi_i$ in Figure~\ref{?}). 
To address the inner loops, our algorithm works with the \emph{control flow graph} (CFG) $H$ of the loop body of the loop $W$ and considers the \emph{execution paths} in this CFG. The CFG $H$ is a directed graph whose vertices are the program counters of the loop body and whose edges describe the one-step jumps between these program counters. Except for the standard semantics of the jumps emitting from assignment statements and conditional branches, for a program counter that represents the entry point of an inner while loop that is not nested in other inner loops, we have the special treatment that the jump at the program counter is directed to the termination program counter of this inner loop in the loop body of $W$ (i.e., skipping the execution of this inner loop). An \emph{execution path} in the CFG $H$ is a directed path of program counters that ends in (i) either the termination program counter of the loop body of $W$ without visiting a  program counter that represents the \textbf{break} statement or (ii) a first \textbf{break} statement without visiting prior \textbf{break} statements. 
An example is as follows.

% Notably, our back-end includes two additional features. The first one is the functionality to remove invalid transitions with unsatisfiable guard condition $\tsGuardcond$. 
% The second one is the treatment of the situation of the unsatisfiability in the application of Farkas' Lemma 
% (see $-1\ge 0$ at the bottom of Figure~\ref{tab:farkasinit} and Figure~\ref{tab:farkascons}), which is however missing in the original tool StInG~\cite{Sting}. The former can simplify the \LTS{} to improve time efficiency and the later can increase accuracy. 
% A key difficulty in the second one is that we obtain polyhedra rather than polyhedral cones, and thus cannot directly apply the generator computation. To address this difficulty, we show that it suffices to consider $\mu=1$ in Figure~\ref{tab:farkascons} and include the generators of both the polytope and the polyhedral cone of the Minkowski decomposition of the polyhedron. As its correctness proof is somewhat technical, we relegate them to Appendix~\ref{sec:appendix_mu1} and Appendix~\ref{sec:appendix_minkowski}.

\begin{example}\label{eg:janne_cfg}
Consider the janne\_complex program from ~\cite{DBLP:conf/vmcai/BoutonnetH19} in Figure~\ref{fig:jannecomplex}.  
The CFG of the program is given in Figure~\ref{fig:janne_cfg} where the nodes correspond to the program counters, the directed edges with guards specifies the jumps and their conditions, and the affine assignments are given in the program counters $A_1,A_2,A_3$. 

\newsavebox{\jannecomplex}
\begin{lrbox}{\jannecomplex}
\begin{lstlisting}[mathescape]
while($x<30$){
    while($y<x$){
        if ($y>5$) $y=y*3$;
        else $y=y+2$;
        if (y>=10 && y<=12) $x=x+10$;
        else $x=x+1$;
    }
    $x=x+2$; $y=y-10$;
}
\end{lstlisting}
\end{lrbox}

\begin{figure}[h]
\centering
\usebox{\jannecomplex}
\caption{The janne\_complex program}
\label{fig:jannecomplex}
\end{figure}

\begin{figure}[h]
\centering
\includegraphics[width=0.8\linewidth]{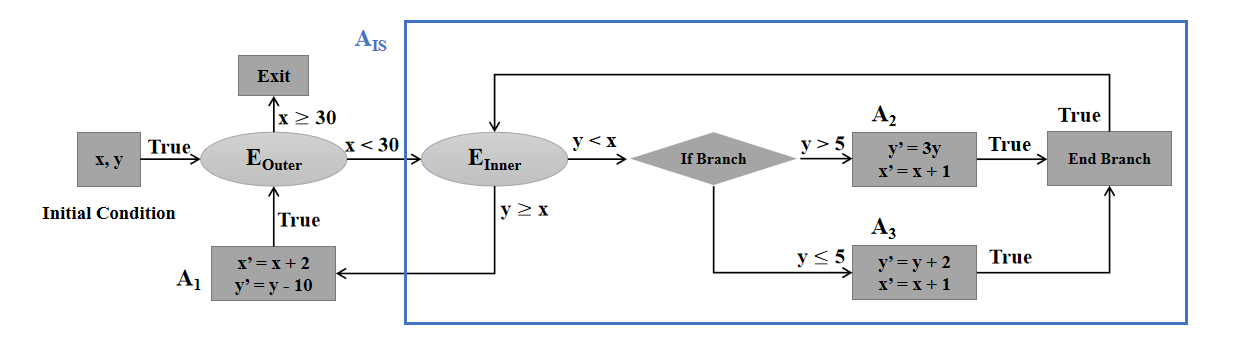}
\caption{The CFG of \emph{janne\_complex}~\cite{DBLP:conf/vmcai/BoutonnetH19}}
\label{fig:janne_cfg}
\end{figure}

We denote by $W$ the outer loop with entry point $E_{\mathrm{Outer}}$, and by $W'$ the inner loop with entry point $E_{\mathrm{inner}}$.    
%the CFG of a specific example with outermost loop $W$ and the inner loop $W'$ within $W$ from \emph{janne\_complex}~\cite{DBLP:conf/vmcai/BoutonnetH19} in Figure~\ref{fig:janne_cfg}. 
The execution path starts at the \emph{Initial Condition} $\left[x, y \right]$, jumps to the next vertices along the edge whose condition is satisfied (e.g., \emph{True} is tautology, \emph{$x < 30$} is satisfied when variable $x$ value is less than $30$, etc.), and terminates in the \emph{Exit} statement. 
The only execution path for the loop body of $W$ is %$E_{Outer} \rightarrow A_{IS} \rightarrow A_{1}$, 
$A_{IS} \rightarrow A_{1}$,
%\rightarrow E_{Outer}$, 
for which we abstract the whole inner loop by $A_{IS}$. 
%The \emph{$E_{Outer}$} means the outermost loop entry of $W$ and \emph{$E_{Inner}$} means the inner-loop entry of $W'$ in $W$. 
%The \emph{$A_1$, $A_2$, $A_3$} represents the assignment statements in the program and \emph{$A_{IS}$} is a special assignment statement for the inner-loop $W'$ which could be obtained by computing loop summary for $W'$. 
\qed 
\end{example}

Based on the CFG $H$ and the execution paths, our approach constructs the \LTS{} for the outer loop $W$ as follows. Since the output of an inner while loop $W'$ in $W$ cannot be exactly determined from the input to the loop $W'$, we first have fresh output variables 
$\overline{x}_{W',1},\dots,\overline{x}_{W',n'}$ to represent the output values of the variables $\overline{x}_{W',1},\dots,\overline{x}_{W',n'}$ after the execution of the inner loop $W'$. These output variables are used to express the loop summaries of these inner loops. 

Then, to get the numerical information from execution paths,  we symbolically compute the values of the program variables at each program counter in an execution path. In detail, given an execution path $\omega=\iota_1,\dots,\iota_k$ where each $\iota_i$ is a program counter of the loop body of the loop $W$, our approach computes the affine expressions $\alpha_{\tsVar,i}$ and PAPs $\beta_i$ (for $\tsVar\in \tsVars$ and $1\le i\le k$) over the program variables in $\tsVars$ (for which they represent their initial values at the start of the loop body of $W$ here) and the fresh output variables. 
%(for the output of the inner loops). 
The intuition is that (i) each affine expression $\alpha_{\tsVar,i}$ represents the value of the variable $\tsVar$ at the program counter $\iota_i$ along the execution path $\omega$ and (ii) each PAP $\beta_i$ specifies the condition that the program counter $\iota_i$ is reached along the execution path $\omega$. The computation is recursive on $i$ as follows.

Denote the vectors $\alpha_i:=(\alpha_{\tsVar_1,i},\dots, \alpha_{\tsVar_n,i})$ and  $\overline{x}_{W'}=(\overline{x}_{W',1},\dots,\overline{x}_{W',n'})$. For the base case when $i=1$, we have $\alpha_1=(\tsVar_1,\dots,\tsVar_n)$ and $\beta_1=\mathbf{true}$ that specifies the initial setting at the start program counter $\iota_1$ of the loop body of the original loop $W$. For the recursive case, suppose that our approach has computed the affine expressions in $\alpha_{i}$ and the PAP $\beta_i$. We classify four cases below:

\begin{itemize}
\item \emph{Case 1:} The program counter $\iota_{i}$ is an affine assignment statement $\mathbf{x}:=\mathbf{F}(\mathbf{x})$. Then we have that $\alpha_{i+1}= \alpha_i[\mathbf{F}(\mathbf{x})/\mathbf{x}]$ and $\beta_{i+1}:=\beta_i$.  
\item \emph{Case 2:} The program counter $\iota_{i}$ is a conditional branch with branch condition $b$ and the next program counter $\iota_{i+1}$ follows its \textbf{then}-branch. Then the vector $\alpha_{i+1}$ is the same as $\alpha_i$, and the PAP $\beta_{i+1}$ is obtained as $\beta_{i+1}=\beta_i\wedge b$.   
\item \emph{Case 3:} The program counter $\iota_{i}$ is a conditional branch with branch condition $b$ and the next program counter $\iota_{i+1}$ follows its \textbf{else}-branch. The only difference between this case and the previous case is that $\beta_{i+1}$ is obtained as $\beta_{i+1}:=\beta_i\wedge \neg b$.   
\item \emph{Case 4:} The program counter $\iota_{i}$ is the entry point of an inner while loop $W'$ of $W$ and $\iota_{i+1}$ is the successor program counter outside $W'$ in the loop body of $W$. Then $\alpha_{i+1}:=\overline{x}_{W'}$ and $\beta_{i+1}:=\ProcSmry_{W'}(\alpha_i, \overline{x}_{W'})$. Here we use the ouput variables to express the loop summary. Note that the loop summary $\ProcSmry_{W'}$ is recursively computed. 
\end{itemize}  

\begin{example}\label{eg:evolution}
Continue with the execution path in Example~\ref{eg:janne_cfg}. 
\begin{figure}[h]
    \begin{footnotesize}
    $$
    \begin{array}{c}
        \alpha_{1}=[x,y],\beta_{1}=\mathbf{true} 
        \quad \underrightarrow{x<30} \quad
        \alpha_{2}=[x,y],\beta_{2}=\beta_{1} \wedge x<30 
        \quad \underrightarrow{A_{IS}} \\
        \alpha_{3}=[\overline{x}_{W'},\overline{y}_{W'}],\beta_{3}=\beta_{2} \wedge \ProcSmry_{W'}(\alpha_{2}, \alpha_{3}) 
        \quad \underrightarrow{A_{1}} \quad
        \alpha_{4}=[\overline{x}_{W'}+2,\overline{y}_{W'}-10],\beta_{4}=\beta_{3}
    \end{array}
    $$
    %}
    \end{footnotesize}
    \caption{The evolution of $\alpha_{i}$ and $\beta_{i}$ for the execution path of $W$ in Figure~\ref{fig:janne_cfg}}
    \label{fig:outer_inner}
\end{figure}
The evolution of $\alpha_{i}$ and $\beta_{i}$ with the initial setting $\alpha_{1}=[x,y],\beta_{1}=\mathbf{true}$ is given in Figure~\ref{fig:outer_inner}.  %where $W$ denotes the current loop. 
\qed 
\end{example}

After the $\alpha_i,\beta_i$'s are obtained for an execution path $\omega=\iota_1,\dots,\iota_k$ from the recursive computation above, we let the PAP $\Psi_\omega:=\bigwedge_{i\in I} \beta_i$ where the index set $I$ is the set of all $1\le i\le k$ such that the program counter $\iota_i$ corresponds to either a conditional branch or the entry point of an inner while loop, and the vector of affine expression $\alpha_{\omega}:=\alpha_{k+1}$. Note that the PAP $\Psi_\omega$ is the condition that the execution of the loop body follows the execution path $\omega$, and the affine expressions in the vector $\alpha_\omega$ represent the values of the program variables after the execution path $\omega$ of the loop body of $W$ in terms of the initial values of the program variables and the fresh variables for the output of the inner while loops in $W$.

Finally, our approach constructs the \LTS{} for the loop $W$ and we only present the main points:
\begin{itemize}
\item First, for each execution path $\omega$ of the loop body of $W$, we have a standalone location $\tsLoc_\omega$ for  this execution path. Recall that we abstract the inner loops, so that the execution paths can be finitely enumerated. 
\item Second, for all locations $\tsLoc_\omega,\tsLoc_{\omega'}$ (from the execution paths $\omega,\omega'$), we have the transition $\tau_{\omega,\omega'}:=(\tsLoc_{\omega}, \tsLoc_{\omega'}, \Psi_{\omega} \wedge \Psi'_{\omega'}\wedge \mathbf{x}'=\alpha_\omega)$ which means that if the execution path in the current iteration of the loop $W$ is $\omega$, then in the next iteration the execution path can be $\omega'$ with the guard condition $\Psi_{\omega} \wedge \Psi'_{\omega'}\wedge \mathbf{x}'=\alpha_\omega$ that comprises the conditions for the execution paths $\omega,\omega'$ and the condition $\mathbf{x}'=\alpha_\omega$ for the next values of the program variables.
\item Third, we enumerate all possible initial locations $l_\omega$, along with their corresponding initial condition $\tsInitcond= G \wedge \Psi_{\omega}$. To derive loop summary, we follow the standard technique (see e.g.~\cite{DBLP:conf/vmcai/BoutonnetH19}) to include the input variables $\tsVars_{\mathsf{in}}$ and conjunct the affine assertion $\bigwedge_{\tsVar\in\tsVars} \tsVar=\tsVar_{\mathsf{in}}$ into each disjunctive clause of the initial condition $\tsInitcond$. 

Manually specified initial condition can also be conjuncted into 
$\tsInitcond$. 
\end{itemize}

A detailed process that handles $\mathbf{break}$ statement is similar to the unnested situation. Again, we can remove invalid transitions by checking whether their guard condition is satisfiable or not. 

At the end of the illustration of our algorithms, we discuss possible extensions as follows. 

\end{document}